\documentclass[journal,twoside,web]{ieeecolor}

\usepackage{generic}
\usepackage{cite}

\usepackage{amsmath,amssymb,amsfonts,amsthm,float}
\usepackage{algorithm}
\usepackage{algpseudocode}
\usepackage{graphicx}
\usepackage{textcomp}
 \usepackage{calc}
 \usepackage{tikz}
\usepackage{cases}
\usepackage{verbatim}
\usepackage{hyperref}
\usepackage{color, xcolor, soul}
\usepackage{mathtools, nccmath}
\usepackage{dsfont}

\theoremstyle{plain}
\newtheorem{thm}{Theorem}

\newtheorem{prop}{Proposition}

\newtheorem{lemma}{Lemma}

\theoremstyle{definition}
\newtheorem{defn}{Definition}

\newtheorem*{rem*}{Remark}

\theoremstyle{remark}

\newcommand{\R}{\mathbb{R}}
\newcommand{\N}{\mathbb{N}}
\renewcommand{\S}{\mathcal{S}}
\newcommand{\ind}[1]{\mathds{1}_{#1}}
\newcommand{\ones}{\mathbf{1}}
\newcommand{\zeros}{\mathbf{0}}
\newcommand{\card}[1]{|#1|}

\renewcommand{\rm}[1]{\mathrm{#1}}

\newcommand{\ac}{\mathcal{A}} 
\newcommand{\p}{\mathcal{I}} 
\newcommand{\W}{\mathrm{W}} 
\newcommand{\U}{\mathrm{U}_i} 
\newcommand{\G}{\mathrm{G}} 
\renewcommand{\k}{\kappa} 

\newcommand{\br}{\text{BR}_i} 
\renewcommand{\ne}{a^{\mathrm{ne}}} 
\newcommand{\NE}{\mathrm{NE}} 
\newcommand{\pob}{\mathrm{Eff}} 
\newcommand{\poa}{\mathrm{PoA}} 
\newcommand{\emp}{a^{\varnothing}}  
\newcommand{\sol}{\mathrm{sol}} 

\newcommand{\rr}{\mathcal{R}} 
\newcommand{\w}{w} 
\newcommand{\f}{u} 
\newcommand{\ww}{\mathcal{W}}
\newcommand{\fw}{\mathcal{U}}  
\newcommand{\setgm}{\mathcal{G}_{\ww, \fw}} 
\newcommand{\setgmn}{\mathcal{G}^{n}_{\ww, \fw}} 
\newcommand{\cc}{\mathrm{C}}  

\newcommand{\fa}{\f_{\poa}} 
\newcommand{\fmc}{\f_{\mathrm{mc}}} 
\newcommand{\wb}{\w^{b}}  
\newcommand{\fb}{\f^{b}}  

\newcommand{\abr}{a^{\mathrm{br}}}  
\newcommand{\aopt}{a^{\mathrm{opt}}}  
\newcommand{\pt}{\mathcal{P}}  
\newcommand{\rrsub}{\rr^{\ell, p}}
\newcommand{\babr}{\mathrm{b}^p}  
\newcommand{\baopt}{\mathrm{o}^p}  
\newcommand{\iabr}{\mathrm{B}^p}  
\newcommand{\iaopt}{\mathrm{O}^p}  
\newcommand{\pbdual}{\beta}  
\newcommand{\pbvar}{\eta_p^{\ell}}  
\newcommand{\lblr}{\ell_r}  

\newcommand{\ay}{y}  
\newcommand{\bz}{z}  
\newcommand{\xx}{x}  
\newcommand{\lamonebeta}{\pbdual_{\ay, \bz}}  
\newcommand{\lamallbeta}{\pbdual_{\lambda}}  
\newcommand{\lamss}{\lambda_{ss}}  
\newcommand{\fmaxj}{M_{\f}^{\ay}} 
\newcommand{\myz}{x} 

\newcommand{\paval}{\theta(\ay, \xx, \bz)}  
\newcommand{\paop}{\Theta(\ay, \xx, \bz)}  
\newcommand{\rraxb}{\rr_{\ay, \xx, \bz}^k} 

\newcommand{\wsc}{\w_{\mathrm{sc}}} 
\newcommand{\ipoa}{\mathcal{X}} 
\newcommand{\fxx}{\f^{\ipoa}} 

\def\BibTeX{{\rm B\kern-.05em{\sc i\kern-.025em b}\kern-.08em
    T\kern-.1667em\lower.7ex\hbox{E}\kern-.125emX}}
\begin{document}

\title{Optimal Utility Design of Greedy Algorithms in Resource Allocation Games}
\author{Rohit Konda, \IEEEmembership{Student Member, IEEE}, Rahul Chandan \IEEEmembership{Student Member, IEEE}, David Grimsman \IEEEmembership{Member, IEEE}, Jason R. Marden \IEEEmembership{Member, IEEE}
\thanks{R. Konda (\texttt{rkonda@ucsb.edu}), R. Chandan, and J. R. Marden are with the Department of Electrical and Computer Engineering at the University of California Santa Barbara and D. Grimsman is with the Department of Computer Science at Brigham Young University. This work is supported by \texttt{ONR grant \#N00014-20-1-2359}, AFOSR grants \texttt{\#FA9550-20-1-0054} and \texttt{\#FA9550-21-1-0203}, and the Army Research Lab through the \texttt{ARL DCIST CRA \#W911NF-17-2-0181}.}}

\maketitle

\begin{abstract}
 Designing distributed algorithms for multi-agent problems is vital for many emerging application domains, and game-theoretic approaches are emerging as a useful paradigm to design such algorithms. However, much of the emphasis of the game-theoretic approach is on the study of equilibrium behavior, whereas transient behavior is often less explored. Therefore, in this paper we study the transient efficiency guarantees of \emph{best response processes} in the context of resource-allocation games, which are used to model a variety of engineering applications. Specifically, the main focus of the paper is on designing utility functions of agents to induce optimal short-term system-level behavior under a best-response process. Interestingly, the resulting transient performance guarantees are relatively close to the optimal asymptotic performance guarantees. Furthermore, we characterize a trade-off that results when optimizing for both asymptotic and transient efficiency through various utility designs.
\end{abstract}


\section{Introduction}
\label{sec:int}

The analysis and control of multi-agent systems has received a significant amount of attention in recent years, due to its tremendous potential for solving various important problems. Some pertinent examples where these systems have found success include wireless communication networks \cite{han2012game}, UAV swarm task allocation \cite{roldan2018should}, news subscription services \cite{hsu2020information}, vaccinations during an epidemic \cite{hota2019game}, facility location \cite{1181966}, coordinating the charging of electric  vehicles \cite{martinez2020decentralized}, and national defense \cite{lee2020perimeter}, among others. The quintessential challenge in designing algorithms for these scenarios is to arrive at well-performing system-level behavior, as measured by some given global objective, that emerges in a distributed and scalable fashion. Therefore, the system designer is tasked to construct distributed algorithms that satisfy locality constraints while optimizing a given global objective.

One classic distributed design is the \emph{greedy algorithm}, where the agents are ordered sequentially, and at each step of the execution, a single agent optimizes the global objective unilaterally given the previous agents' decisions in the sequence. In general, the greedy algorithm is not guaranteed to find the globally optimal solution, but is a quick and elegant way to derive an approximately optimal solution. In fact, in certain well-structured domains, there may even be provable guarantees on the approximation ratios. For example, greedy algorithms are known to have approximation guarantees in set covering problems \cite{chvatal1979greedy}, $k$-extendible problems \cite{mestre2006greedy}, submodular maximization problems \cite{nemhauser1978analysis}, etc.

A natural extension to the greedy algorithm is to repeatedly allow agents to respond dynamically to other agent's decisions, where, in a progressive fashion, a single, chosen agent optimizes their local objective (utility) function at each step. Here, the agent-specific local objective functions may not be equal to the global objective, and in fact, may be desirable for them to not be equivalent \cite{gairing2009covering, paccagnan2018utility}. We can study the asymptotic behavior of the possible iterative algorithms by examining the resulting equilibrium using a game-theoretic approach. When these equilibrium correspond to the set of Nash equilibrium, the resulting approximation ratio is also known as the \emph{price of anarchy} \cite{koutsoupias1999worst}. Many previous works are dedicated to analyzing the resulting price of anarchy guarantees in many different settings (for e.g., see \cite{roughgarden2009intrinsic, vetta2002nash}), and in some instances, the price of anarchy can even be computed through a convex optimization program \cite{roughgarden2009intrinsic}.

While the resulting approximation guarantees of the game-theoretic approach are positive, these guarantees only emerge asymptotically. In fact, arriving at Nash equilibrium may even take an exponential time \cite{fabrikant2004complexity}, rendering the resulting approximation guarantees irrelevant in many realistic multi-agent scenarios. For example, there may be an extremely large number of agents in the multi-agent scenario or the relevant situational parameters may be time-varying and volatile or there may be computational and run-time restrictions on the agents. In these instances, expecting that the agents will converge to Nash equilibrium may not be a reasonable assumption.

Therefore, in this paper, we shift focus to the transient behavior of these game-theoretic algorithms. In this sense, we extend the greedy algorithm through more sophisticated utility function designs but also consider limitations on the time complexity of the iterative process. Explicitly, we benchmark the iterative process to be the \emph{$\k$ round-robin best response algorithm} and study the approximation guarantees that result from various designs of utility functions in the context of the well-studied class of \emph{resource allocation games}. This class of games can model many relevant engineering applications \cite{marden2013distributed} and have convergence guarantees of the best response algorithm to Nash equilibrium. The main results of this paper is as follows.  
\begin{itemize}
    \item The first main result of the paper, in Theorem \ref{thm:optLP}, focuses on tightly characterizing the optimal (under a natural class of utility designs) performance guarantees of the one-round best response algorithm, in which each agent is allowed to alter its decision only once. We also show in Theorem \ref{thm:oneroundC} that the performance guarantees that are possible in this situation are significantly closer, than the greedy algorithm, to the asymptotically optimal performance guarantees, in which the best-response process is run ad infinitum (being at worst $\sim 80\%$ within the asymptotically optimal performance guarantee and at most $\sim 13\%$ better than the greedy algorithm).
    \item Additionally, we show in Theorem \ref{thm:kroundC} that if we shift attention to the $\k > 1$ round robin best response algorithm, where agents are instead allowed to alter their decisions more than once, the overall performance guarantees do not increase. This suggests minimal benefit in performance when allowing agents to alter their decision in response to other agents.
    \item Lastly, we characterize the trade-offs that result from optimizing transient and asymptotic performance guarantees. Interestingly, we show that optimizing for asymptotic performance guarantees can result in arbitrarily bad transient performance guarantees, as shown in Theorem \ref{thm:submodtrade}. Furthermore, a Pareto-optimal frontier is delineated for the asymptotic and transient performance guarantees in Theorem \ref{thm:poapobtradeoff}. 
\end{itemize}
The results of this work suggests an acute diminishing return on the achievable performance guarantees as agents make successive locally optimal decisions. While the current literature is sparse on characterizing short-term behavior of game-processes, the positive results in this work indicate the importance of analyzing the interaction between local objective designs, asymptotic performance, and transient performance as a whole.

Moreover, this work belongs to a larger research trend that aims to study game theoretic models beyond their respective equilibrium. In contrast to the traditional game-theoretic approach, the \emph{game dynamics} are embraced as a valuable feature of the game, where a rigorous study of actualized play can provide important insights about the game model (see, for e.g., \cite{goemans2005sink, candogan2011flows, papadimitriou2019game}). Furthermore, characterizing performance guarantees along these game dynamics is valuable to understanding the transient behavior of the agents. In contrast to the study of equilibrium quality, the literature on transient guarantees is much less developed. However, we highlight an important subset of works that characterize transient performance guarantees in different game-theoretic contexts: such problem domains include affine congestion games \cite{christodoulou2006convergence, bilo2009performances, fanelli2008speed, bilo2018unifying}, market sharing games \cite{christodoulou2006convergence, mirrokni2004convergence}, basic utility games \cite{mirrokni2004convergence}, series-parallel networks, and load-balancing games \cite{suri2007selfish, caragiannis2006tight}. Furthermore, we reference the related literature (see \cite{leme2012curse, correa2015curse, de2014sequential, angelucci2015sequential}) of \emph{price of sequential anarchy} that characterize the efficiency of equilibrium that have a dynamical flavor. While similar in spirit to the previous literature, this work extends beyond characterization to study optimal utility designs that maximize the short-term guarantees that can be achieved. To the author's knowledge, there has been no previous work that addresses optimal short-term performances. Additionally, we also delve into the trade-offs that result from maximizing the transient performance guarantees versus the asymptotic guarantees with regards to a natural class of game dynamics.

\begin{figure}[ht]
    \centering
    \includegraphics[width=250pt]{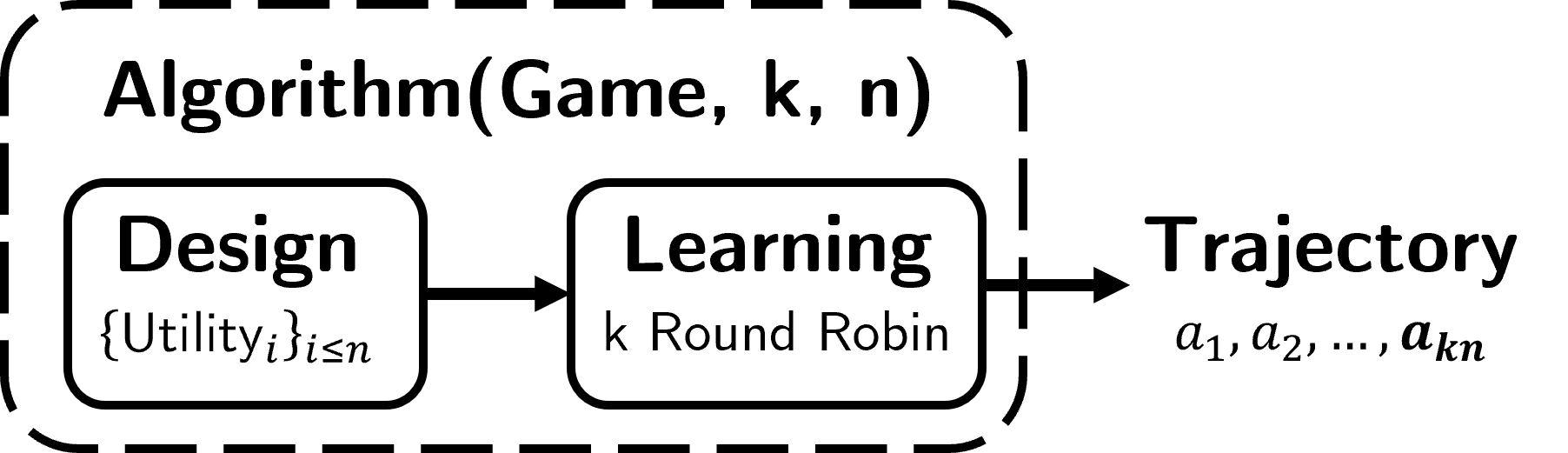}
    \caption{If a given multi-agent scenario with $n$ agents is modeled as a game, the construction of distributed algorithms can be decoupled into two domains: the design of local objectives (utilities) and the design of the learning dynamics. In this paper, we fix the dynamics to the classical round-robin best response and study the effects of the utility design on the resulting efficiency bounds on the trajectory. Moreover, we characterize the efficiency guarantee as the number iterations $\k$ that the Algorithm is run for increases. }
    \label{fig:intro}
\end{figure}

\section{Model}

This work considers multi-agent scenarios in the form of \emph{resource allocation games} \cite{marden2013distributed}. Resource allocation games are characterized by a finite set of resources $\rr = \{r_1, \dots, r_{c})$ that can be utilized by a set of agents $\p = \{1, \dots, n\}$. Each agent $i \in \p$ has a finite action set $\ac_i = \{a^1_i, \dots, a_i^d\} \subseteq 2^\rr$ representing the decisions available to each agent. The performance of a joint action $a = (a_1, \dots, a_n) \in \ac = \ac_1 \times \cdots \times \ac_n$ comprised of every agent's actions is evaluated by a system-level objective function $\W: \ac \to \R_{> 0}$ of the form
\begin{equation}
    \W(a) = \sum_{r \in \bigcup_i a_i}{\w_r(|a|_r)},
\end{equation}
where $|a|_r = \card{\{i \in \p : r \in a_i\}}$ is the number of agents that choose resource $r$ in action profile $a$, and the \emph{welfare rule} $\w_r :\N \to \R_{>0}$ defines the resource-specific welfare determined by the utilization of $r$ by $|a|_r$ agents. System welfares of this form can model several important classes of problems including routing over a network \cite{roughgarden2009intrinsic}, vehicle target assignment \cite{murphey2000target}, sensor coverage \cite{gairing2009covering} etc. Thus the goal is to coordinate the agents to a joint action that maximizes the system welfare, i.e. $\aopt \in \arg \max_{a \in \ac} \W(a)$.

To achieve near-optimal joint actions in a distributed fashion, a game-theoretic approach is utilized where the system designer ascribes the preference structure to each agent through its utility function $\U : \ac \to \R$ to influence its decision making process. In general, these utility rules could either be inherited (e.g., congestion or latency in a given transportation network) or designed (e.g., congestion with monetary incentives). We consider utility functions of the form
\begin{equation}
\label{eq:utildef}
    \U(a_i, a_{-i}) = \sum_{r \in a_i} \f_r(|a|_r),
\end{equation}
where the \emph{utility rule} $\f_r : \N \to \R_{\geq 0}$ defines the resource-specific agent utility determined by $|a|_r$. We use $a_{-i} = (a_1, \dots, a_{i-1}, a_{i+1}, \dots a_n)$ to denote the joint action without the action of agent $i$. The choice of the utility rules $\f_r$ influences the agents' individual behaviors and in turn the emergent collective behavior. Therefore, the goal of this paper is to identify the utility rules $\f_r$ that lead to the most efficient collective behavior.

Given the design above, we will express a given resource allocation game by the tuple $\G = \{\p, \ac, \rr, \{\w_r\}_{r \in \rr}, \{\f_r\}_{r \in \rr}\}$.  In most scenarios of interest, a system designer is required to specify the utility rules $\f_r$ without specific knowledge of the resource allocation game parameters, such as the number of agent $\p$ or the action set $\ac$. To that end, let $\ww$ be the set of possible welfare rules that could be associated with any resource, i.e., $\w_r \in \ww$ for all  $r \in \rr$.  Here, the system designer is tasked with associating a utility rule to each type of resource, i.e.,  the utility rule for any resource $r\in \rr$ with the welfare rule $\w_r$ is of the form $\f_r = \fw(\w_r)$  where $\fw: \ww \to \R^{\N}_{\geq 0}$ is termed the utility design. Lastly, we define the set of resource allocation games that are induced by $\ww$ and $\fw$ as $\setgm$, where a game $\G \in \setgm$ if $\w_r \in \ww$ and that $\f_r = \fw(\w_r)$ for all resources $r \in \rr$.

The central focus of this paper is to understand how the choice of utility rules, derived from $\fw(\cdot)$, impacts the efficacy of the emergent collective behavior. To benchmark the learning process of the agents, we restrict attention to the class of \emph{(round-robin) best response processes}, where a certain agent (out of $n$ agents) performs a best response in a round-robin fashion. For a given joint action $\bar{a} \in \ac$, we say the action $\hat{a}_i$ is a \emph{best response} for agent $i$ if 
\begin{equation}
\hat{a}_i \in \ \br(\bar{a}_{-i}) = \arg \max _ {a_i \in \ac_i} \U(a_i, \bar{a}_{-i}),
\end{equation}
where $\br$ may be non-unique. We also assume that the best response process begins with none of resources being utilized by any of the agents, denoted by the agents selecting the null joint action $\emp := \varnothing$ at time $0$. The underlying algorithm is formalized in Algorithm \ref{alg:cap2}.
\begin{algorithm}
\caption{Best Response Process}
\label{alg:cap2}
\begin{algorithmic}
\Require $a(0) \gets \emp$, $\tau \gets 0$, $T$
\While{$\tau \leq T$}
\State $i \gets (\tau + 1) \mod n$ \Comment{Next agent is selected.}
\State $a_i(\tau+1) \gets \hat{a}_i \in \br(a_{-i}(\tau))$ 
\Comment{$i$ best responds.}
\State $a_{-i}(\tau+1) \gets a_{-i}(\tau)$ \Comment{No other agent moves.}
\State $\tau \gets \tau + 1$
\EndWhile
\end{algorithmic}
\end{algorithm}

\noindent To arrive at non-trivial efficiency guarantees, we define a \emph{$\k$-round walk} as a best response process in which Algorithm \ref{alg:cap2} is run for $T = \k \cdot n$ steps. Here, the set of agents perform best responses in succession $\k$ times. We assume that during a $\k$-round walk, agent $i$ selects its best response $a_i(\tau+1)$ arbitrarily from $\br(a_{-i}(\tau))$ if it is not unique. This induces a set of possible action trajectories of the form $(a(0)=\emp, a(1) \dots, a(\k n - 1 ), a(\k n))$ selected by the agents throughout the best response process. The potential solution set that occurs after the agents run a $\k$-round walk is denoted by $\sol(\k) \subset \ac$ with
\begin{equation*}
\sol(\k) = \{a(\k n) \text{ for each trajectory starting at } a(0) = \emp\}.
\end{equation*}
The worst achievable efficiency at the end of the $\k$-round walk with respect to the best achievable system welfare is defined by the following competitive ratio
\begin{equation}
\label{eq:effG}
\pob(\G; \k) = \frac{\min_{a \in \sol(\k)}{\W(a)}}{\max_{a \in \mathcal{A}}{\W(a)}} \in [0, 1],
\end{equation}
where a ratio closer to $1$ implies the worst case efficiency after $\k$ rounds is closer to optimal. We additionally extend the efficiency measure to a set of games $\mathcal{G}$ as
\begin{equation}
\label{eq:effsetG}
    \pob(\mathcal{G}; \k) = \inf_{\G \in \mathcal{G}} \pob(\G; \k).
\end{equation}

The directive of this paper is to derive agent utility rules that optimize the efficiency guarantees of the resulting collective behavior. More specifically, for a given class of welfare rules $\ww$ and number of rounds $\k \geq 1$, the main goal is to characterize the optimal performance guarantees of the form
\begin{equation}
    \pob^*(\ww;\k) = \sup_{\fw: \ww \to \R^{\N}_{\geq 0}} \pob(\setgm; \k).
\end{equation}
Specifically, this paper will seek to address how these optimal efficiency guarantees change as a function of the set of possible welfare rules $\ww$ as well as the number of rounds $\k$.

\section{One Round Guarantees}

Our first set of results characterizes the attainable performance guarantees for just a single round of the best response process. We focus on the performance of the one-round walk to describe the quickest nontrivial guarantees that can occur under the round-robin best response process, as each agent is required to perform only one best response to arrive at the resulting joint action $a \in \sol(1)$. Thus, we derive the one-round guarantees through a linear program construction that is a function of both the set of allowable welfare rules $\ww$ and the utility rules $\fw$. Moreover, we will restrict attention to welfare rules that are generated by the span of a given set of basis welfare rules, as defined below.

\begin{defn}[Basis]
The set of welfare rules $\ww$ is spanned from a set of basis welfare rules $\{\w^1, \dots, w^m\}$ \footnote{We assume, without loss of generality, that $\w^j(1)=1.$} if $\ww$ can be characterized by any non-negative combination of the basis welfare rules $\{\w^1, \dots, w^m\}$, i.e., 
\begin{equation}
    \ww = \{w: w=\sum_{\ell=1}^m \alpha^\ell \w^\ell, \ \forall \alpha^1, \dots, \alpha^m \geq 0\}.
\end{equation}
\end{defn}

We remark that the performance guarantees are identical regardless of if we consider the welfare set $\ww = \{\w^1, \dots, w^m\}$ or consider $\ww$ to be the set spanned from the basis $\{\w^1, \dots, w^m\}$. Furthermore, we will restrict attention to welfare rules that are \emph{submodular}, or informally, welfare rules that admit a notion of decreasing marginal-returns that are commonplace in many objectives relevant to engineered systems. 

\begin{defn}[Submodularity]
A welfare rule $\w$ is submodular if $\w(j+1) \geq \w(j)$ for all $j \in \N$ and $\w(j+1) - \w(j)$ is non-increasing in $j$.  
\end{defn}

\noindent The first main contribution is stated below, where we characterize the best achievable efficiency guarantees for a one-round walk through a utility design $\fw$.

\begin{thm}
\label{thm:optLP}
Suppose that $\ww$ is spanned from a set of welfare rules $\{\w^1, \dots, \w^m\}$ where each $w^\ell$ is submodular. Then the optimal efficiency guarantees achievable with a one-round best response process is given by
\begin{equation}
\label{eq:effoptratio}
    \pob^*(\ww; 1) = \min_{1 \leq \ell \leq m} \frac{1}{\pbdual^\ell}
\end{equation}
where $\beta^{\ell} \in \R_{\geq 0}$ is the solution to the following program.
\begin{align}
(u^{\ell}, \pbdual^{\ell}) \ &\in \ \arg \min_{\pbdual, \f \in \R^{\N}_{\geq 0}} \quad \pbdual \quad \textrm{subject to:}  \label{eq:submodLPoptimal} \\
\pbdual \w^{\ell}(\ay) &\geq \sum_{i=1}^\ay \f(i) - \bz \f(\ay+1) + \w^{\ell}(\bz) \quad \forall \ay, \bz \geq 1, \nonumber
\end{align}
where we take $\f^{\ell}(1) = 1$ and $\ay, \bz \in \N$. Furthermore, a utility design $\fw$ that achieves this optimal efficiency guarantee is linear and of the form
\begin{equation}
    \label{eq:optlinear}
    \fw\left(w=\sum_{j=1}^m \alpha^{\ell} \w^{\ell}\right) = \sum_{j=1}^m \alpha^{\ell} u^{\ell},
\end{equation}
where $u^{\ell}$ is the corresponding solution in Eq. \eqref{eq:submodLPoptimal}.
\end{thm}

The above theorem sets forth a prescriptive process by which to characterize the optimal efficiency guarantees of a one-round best response process. \footnote{While we mostly focus on the class of submodular welfare rules in this paper, the linear program in Eq. \eqref{eq:submodLPoptimal} can be extended to consider other classes.} Acquiring $\pbdual^{\ell}$ through the program in Eq. \eqref{eq:submodLPoptimal} may be computationally infeasible in general; however, by considering certain structured classes of welfare rules, we can derive closed form expressions for the one round performance guarantees. We therefore consider a natural restriction of submodular welfare rules centered around the idea of \emph{curvature}, which is defined below.

\begin{defn}[Curvature]
A submodular welfare rule $\w$ has a curvature of $\cc \in [0, 1]$ if $\cc = 1 - \lim_{n \to \infty} (\w(n+1) - \w(n))/\w(1)$.
\end{defn}

\noindent In this sense, curvature characterizes the rate of diminishing returns associated with a welfare rule $\w$.With this, we can arrive at a tight, closed-form characterization of the optimal one-round performance guarantees, as shown below.

\begin{thm}
\label{thm:oneroundC}
Let the set $\ww$ comprise of all welfare rules $\w$ such that $\w \in \R_{> 0}^{\N}$ has a curvature of at most $\cc$. Then the optimal efficiency guarantees achievable with a one-round best response process satisfies
\begin{equation}
    \label{eq:oneeqC}
    \pob^*(\ww; 1) = 1 - \frac{\cc}{2} \left( \geq \frac{1}{2} \right).
\end{equation}
\end{thm}

\noindent The optimal utility design that achieves the above efficiency guarantee also has a closed form expression, which is found in the Appendix. The results in Theorem \ref{thm:oneroundC} suggests that, under the optimal utility design, running best response processes for these classes of games can result in the agents coordinating to a high quality joint action very quickly. If the curvature $\cc$ is close to $0$, we can even arrive at an approximation guarantee of nearly $1$ after only a one-round walk.

We frame the results of Theorem \ref{thm:oneroundC} with respect to the performance of the \emph{greedy algorithm}, a classical approach to submodular optimization problems \cite{nemhauser1978analysis}. In the greedy algorithm, each agent $i$ computes its greedy solution $a_i^{\rm{gr}}$ in sequence as the argument to the following optimization problem.
\begin{equation}
\label{eq:greedydef}
    a_i^{\rm{gr}} \in \arg \max_{a_i \in \ac_i} \W(a_1^{\rm{gr}}, \dots, a_{i-1}^{\rm{gr}}, a_i, \emp_{i+1}, \dots \emp_{n}),
\end{equation}
arriving at a greedy joint action $a^{\rm{gr}}$. In fact, running the greedy algorithm is equivalent to directly assigning the utilities in Eq. \eqref{eq:utildef} as $\U(a) = \W(a)$ for all $i \in \p$ and $a \in \ac$ and running a round-robin best response process for one round. We refer to this non-local utility design as the \emph{common interest} ($\rm{CI}$) design. Now, we characterize the efficiency guarantees of the greedy algorithm with respect to curvature below.
\begin{prop}
\label{prop:effMCC}
Let the set $\ww$ comprise of all welfare rules $\w$ such that $\w \in \R_{> 0}^{\N}$ has a curvature of at most $\cc$. Then the efficiency guarantee associated with the greedy algorithm in Eq. \eqref{eq:greedydef} is
\begin{equation}
\label{eq:curvGR}
    \pob(\mathcal{G}_{\ww, \rm{CI}}; 1) = (1 + \cc)^{-1} \left( \leq 1 - \frac{\cc}{2} \right),
\end{equation}
where $\mathcal{G}_{\ww, \rm{CI}}$ is the set of games induced by $\ww$ and the common interest utility $\rm{CI}$, where $\U(a) = \W(a)$ for all $i \in \p$ and $a \in \ac$.
\end{prop}

The efficiency guarantee in Eq. \eqref{eq:curvGR} actually exactly matches the bound given in the submodular optimization literature for general submodular set functions \cite{conforti1984submodular}. A visual comparison between the performance guarantees of the greedy algorithm to the optimal one-round utility design is depicted in Figure \ref{fig:efffront}. Additionally, we depict the optimal asymptotic guarantees, where we allow the best response process to converge. The resulting asymptotic guarantees actually match the best approximation
guarantee, that is $1 - \cc/e$, that can be achieved by any polynomial time algorithm \cite[Theorem 1]{chandan2021tractable}, \cite[Theorem 2]{paccagnan2021utility}. We see the resulting loss in performance as we go from asymptotic time to the one-round walk to the greedy algorithm.

\begin{figure}[ht]
    \centering
    \includegraphics[width=250pt]{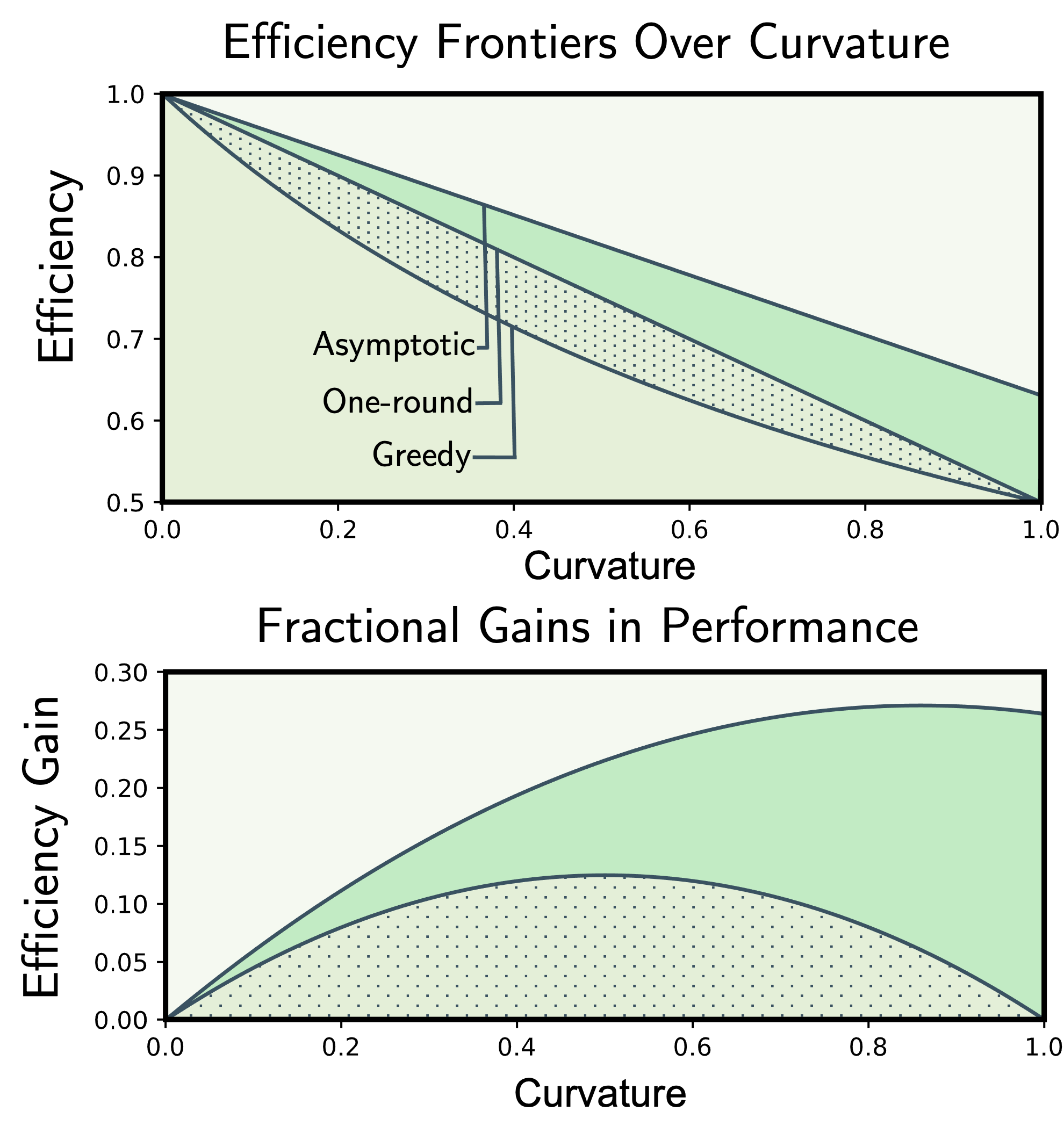}
    \caption{In the top figure, we visually depict the efficiency guarantees of Theorem \ref{thm:oneroundC} and Proposition \ref{prop:effMCC} with respect to the optimal asymptotic guarantees. Additionally, the fractional gains in the performance when moving from the greedy solution to the optimal one-round and the asymptotic solutions are depicted in the bottom figure.}
    \label{fig:efffront}
\end{figure}

\section{Multiple Round Guarantees}

We now extend to $\k$-round walks, and study the resulting efficiency guarantees. Allowing the best response process to continue for more than $\k = 1$ rounds may appear to be a natural avenue to increase the performance guarantees. However, in the next theorem, we show that further rounds do not increase the relative efficiency guarantees. Specifically, with regards to the set of welfare rules of a certain curvature, we derive an upper bound for the efficiency of $\k$-round walk that exactly matches the efficiency guarantee of the one-round walk.

\begin{thm}
\label{thm:kroundC}
Let the set $\ww$ comprise of all welfare rules $\w$ such that $\w \in \R_{> 0}^{\N}$ has a curvature of at most $\cc$. Then the efficiency guarantees of the optimal utility design and the common interest utility with a $\k$-round best response process, for any $\k \geq 1$, is respectively upper bounded by
\begin{align}
    \label{eq:klessC}
    \pob^*(\ww; \k) &\leq 1 - \frac{\cc}{2} \\[1ex]
    \label{eq:klessCCI}
    \pob(\mathcal{G}_{\ww, \rm{CI}}; \k) &\leq (1 + \cc)^{-1},
\end{align}
where the common interest utility $\rm{CI}$ satisfies $\U(a) = \W(a)$ for all $i \in \p$ and $a \in \ac$.
\end{thm}

Notably, for any curvature $\cc \in [0, 1]$, the upper bound in Eq. \eqref{eq:klessC} exactly matches the characterization in Eq. \eqref{eq:oneeqC}, and likewise for the upper bound in Eq. \eqref{eq:klessCCI} and the characterization in Eq. \eqref{eq:curvGR}. Therefore, in regards to the efficiency guarantees, running the best response process for more than one round does not increase the performance. We also remark that the results in Theorem \ref{thm:kroundC} is not endemic to the specific dynamics we consider in this paper. Allowing for different order of play apart from round-robin does not affect the resulting upper bounds. This is further elaborated on in the Appendix. Therefore, in general, this suggests stark diminishing returns for running the best response process for further rounds.

\section{Tradeoffs of Short/Long Term Performance}

So far, we have described the optimal transient performance guarantees through an intelligent utility design. However, we are also interested in the downstream effects of utility designs when running the round-robin best response process asymptotically and the interplay between the achievable short and long term efficiency guarantees. To start, we importantly observe that, since we consider resource allocation games \footnote{Resource allocation games are isomorphic to potential games \cite{monderer1996potential}.}, the set of limit points of the best response process $\lim_{\k \to \infty} \sol(\k)$ must be a subset of the set of \emph{Nash equilibrium} $\NE$ of the game, if the Nash equilibrium are strict. Thus, we explore the resulting efficiency guarantees of the Nash equilibrium to characterize the asymptotic behavior of the best response process. We consider $\ne \in \NE$ to be a Nash equilibrium if
\begin{equation}
    \ne_i \in \arg \max_{a_i \in \ac_i} \U(a_i, \ne_{-i}) \ \text{ for all } i \in \p,
\end{equation}
where $\ne$ is strict if for all $i$, the action $\ne_i$ is the unique maximizer. Then considering the solution set to be $\NE$ as opposed to $\sol(\k)$ for the game $\G$, we arrive at the familiar metric of \emph{price of anarchy} as follows.
\begin{equation}
    \poa(G) = \frac{\min_{a \in \NE}{\W(a)}}{\max_{a \in \mathcal{A}}{\W(a)}}.
\end{equation}
We similarly define $\poa(\mathcal{G}) = \inf_{\G \in \mathcal{G}} \poa(\G)$ mirroring Eq. \eqref{eq:effsetG}. The price of anarchy is a well understood metric, with a host of results on its characterization, complexity, and design \cite{roughgarden2009intrinsic}. In the following theorem, we proceed to identify the reciprocal guarantees for both the utility design that optimizes the transient guarantees of the one round walk and the utility design that optimizes the long-term guarantees of the Nash equilibrium.

\begin{thm}
\label{thm:submodtrade}
Suppose that $\ww$ is the set of all possible submodular welfare rules. Also, let $\fw_{\poa}$ be the utility design that maximizes the price of anarchy and $\fw_1^*$ be the utility design that maximizes the efficiency guarantees of the one-round walk. Then the efficiency guarantees with a one-round best response process for both utility designs are
\begin{equation}
    \pob(\mathcal{G}_{\ww, \fw_{\poa}}; 1) = 0 \quad \quad \pob(\mathcal{G}_{\ww, \fw_1^*}; 1) = \frac{1}{2}. \label{eq:tradfrontone}
\end{equation}
Furthermore, the price of anarchy guarantees of both utility designs are respectively
\begin{equation}
    \poa(\mathcal{G}_{\ww, \fw_{\poa}}) = 1 - \frac{1}{e} \quad \quad \poa(\mathcal{G}_{\ww, \fw_1^*}) = \frac{1}{2}. \label{eq:tradfrontpoa}
\end{equation}
\end{thm}

We observe that while the asymptotic guarantees of $\fw_1^*$ are equivalent to the corresponding transient guarantees, the transient guarantees of $\fw_{\poa}$ unexpectedly degrade to $0$. Interestingly, optimizing for asymptotic performance does not necessarily translate to good transient performance in our setting. To clarify the interplay between the transient and asymptotic guarantees, we would like to characterize the exact Pareto optimal frontier for these guarantees. While calculating this trade-off frontier is difficult to do in general, we restrict our analysis to the specific subset of resource allocation games known as \emph{set covering games} \cite{gairing2009covering} to arrive at an exact trade-off curve. Set covering games are characterized by the following welfare rule.
\begin{equation}
\label{eq:wscdef}
    \wsc(j) = \left\{\begin{array}{lr}
        1, & \text{for } j \geq 1\\
        0, & \text{for } j = 0\\
        \end{array}\right\}.
\end{equation}

\noindent With this, we arrive at the following Pareto frontier characterization, depicted in Figure \ref{fig:tradepareto}. Note that the end points of the trade-off curve matches Eq. \eqref{eq:tradfrontone} and Eq. \eqref{eq:tradfrontpoa} exactly.

\begin{thm}
\label{thm:poapobtradeoff}
Let $\ww = \{\wsc\}$, where $\wsc$, defined in Eq. \eqref{eq:wscdef}, is the set covering welfare rule and $\fw = \{\f\}$ is the corresponding utility rule. Under the constraint that
$\poa(\mathcal{G}_{\wsc, \f}) = Q \in [\frac{1}{2}, 1 - \frac{1}{e}]$, the optimal $\max_{\f} \pob(\mathcal{G}_{\wsc, \f}; 1)$ is
\begin{equation}
\label{eq:tradeoffsetcov}
\left[ \sum^\infty_{j=0} \max \left\{ j! (1 - \frac{1 - Q}{Q} \sum^j_{\tau=1}{\frac{1}{\tau!}}) , 0\right\}+1 \right]^{-1}.
\end{equation}
\end{thm}

\begin{figure}[ht]
    \centering
    \includegraphics[width=250pt]{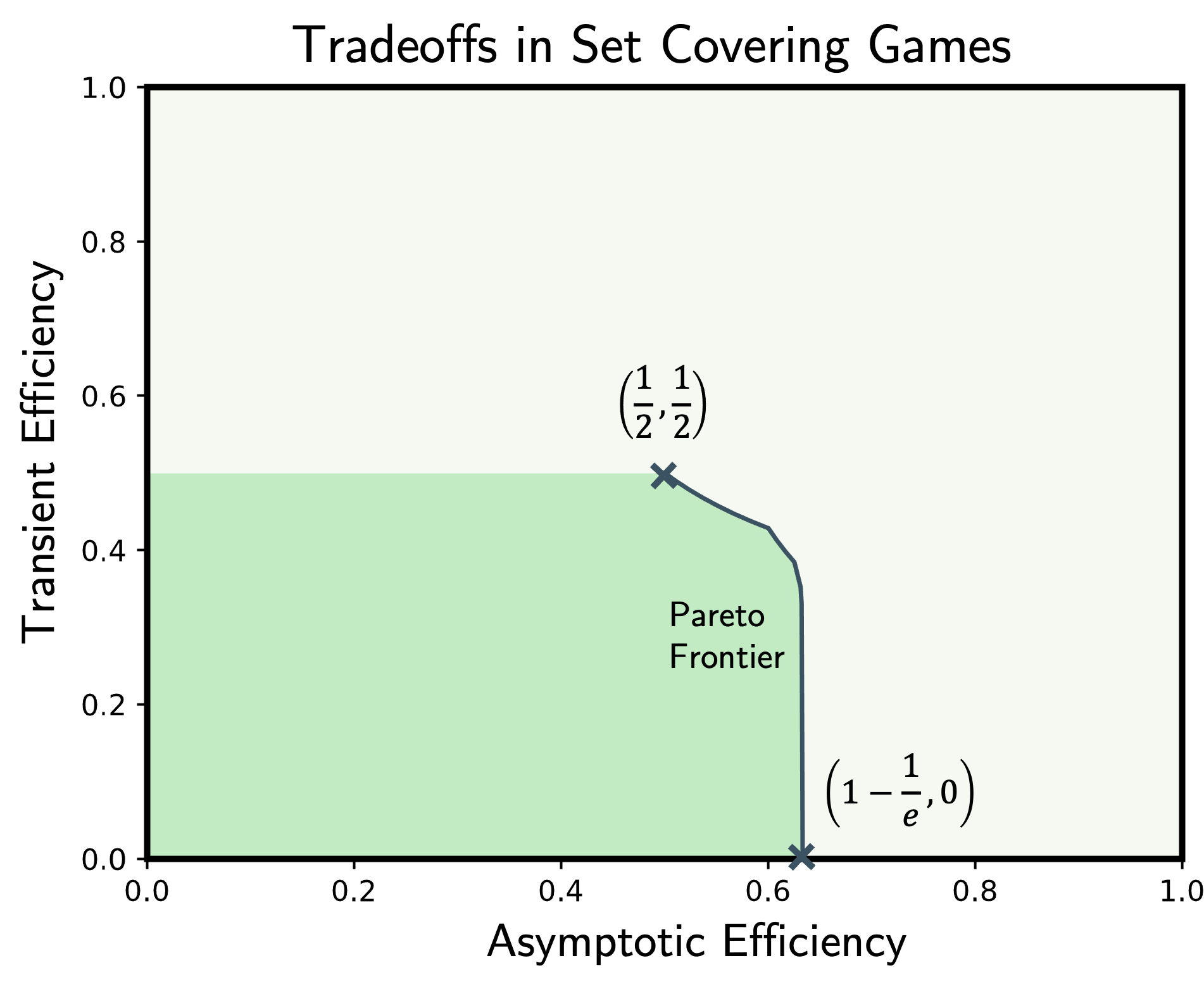}
    \caption{We depict the Pareto-optimal frontier of the one-round efficiency ($\pob(\mathcal{G}_{\wsc, \f}; 1)$) versus the asymptotic efficiency guarantees ($\poa(\mathcal{G}_{\wsc, \f})$) that are possible with regards to the class of set-covering games. We note that the severe drop off in transient efficiency that results from optimizing the asymptotic efficiency.}
    \label{fig:tradepareto}
\end{figure}

Notably in Figure \ref{fig:tradepareto}, we see stark drop-off in transient guarantees, when the price of anarchy is close to $1 - \frac{1}{e}$. This extreme trade-off should prompt a more careful interpretation of asymptotic results, especially in the setting of resource allocation games.

\section{Illustrative Example}
\label{sec:illex}

In this section, we examine the problem of wireless sensor coverage (see \cite{huang2005coverage}) as an illustration of a possible multi-agent scenario. Consider a group of sensors that can observe the region $\rr$. Each sensor has the ability to sense a subset of the region depending on its orientation, physical placement of the sensor, etc. The choice of these parameters constitute the action set $\ac_i$ for each sensor, where a resource is monitored by the sensor if $r \in a_i$. As a whole, the set of sensors wish to arrive at a configuration that maximizes the likelihood of detecting an event. As such, the system welfare is
\begin{equation}
    \W(a) = \sum_{r \in \rr} p_r \cdot (1 - (1 - \mathrm{D})^{|a|_r}),
\end{equation}
where $p_r \in [0, 1]$ indicates the probability that the event will occur at $r$ and $\mathrm{D}$ refers to the conditional probability that the event will be detected by a single sensor given that the sensor is monitoring $r$ and an event does indeed occur at $r$. 

To illustrate the results of the paper, we examine the average performance over $5$ rounds of round-robin best response using the common interest utility design, the asymptotically optimal utility design, and the utility design that optimizes the one-round efficiency across of set of random instances of wireless sensor coverage. Specifically, we simulate $100$ random instances with $20$ agents with $D = .5$. In each simulation, there is a set of $30$ resources that can be covered by agents with the $p_r$ being uniformly selected from the unit interval $[0, 1]$ and subsequently normalized to create a probability distribution $\{p_r\}_{r \in \rr}$. Each agent has $2$ actions available, in addition to the empty allocation $\emp$. Each action $a$ is a consecutive selection of $2$ resources chosen randomly from the resource set $\rr$. 

The resulting system welfare across $5$ rounds for each utility design is highlighted in Figure \ref{fig:numerical}, where the distributions of the system welfare across the randomized instances are depicted with a box and whisker plot. Note that the optimal allocation may also not achieve a $100\%$ detection rate. In Figure \ref{fig:numerical}, we see that worst instance of the optimal one-round performs better than the greedy and asymptotically optimal utility designs when $\k=1$. This is supported in the worst-case analysis presented in this paper. Additionally, we note that the resulting efficiency plateaus quickly, with almost no differences in efficiency after two rounds of best response - confirming that successive rounds give diminishing returns in system performance. Interestingly, on average, the differences in performance across utility designs is much more subtle.

\begin{figure}[ht]
    \centering
    \includegraphics[width=250pt]{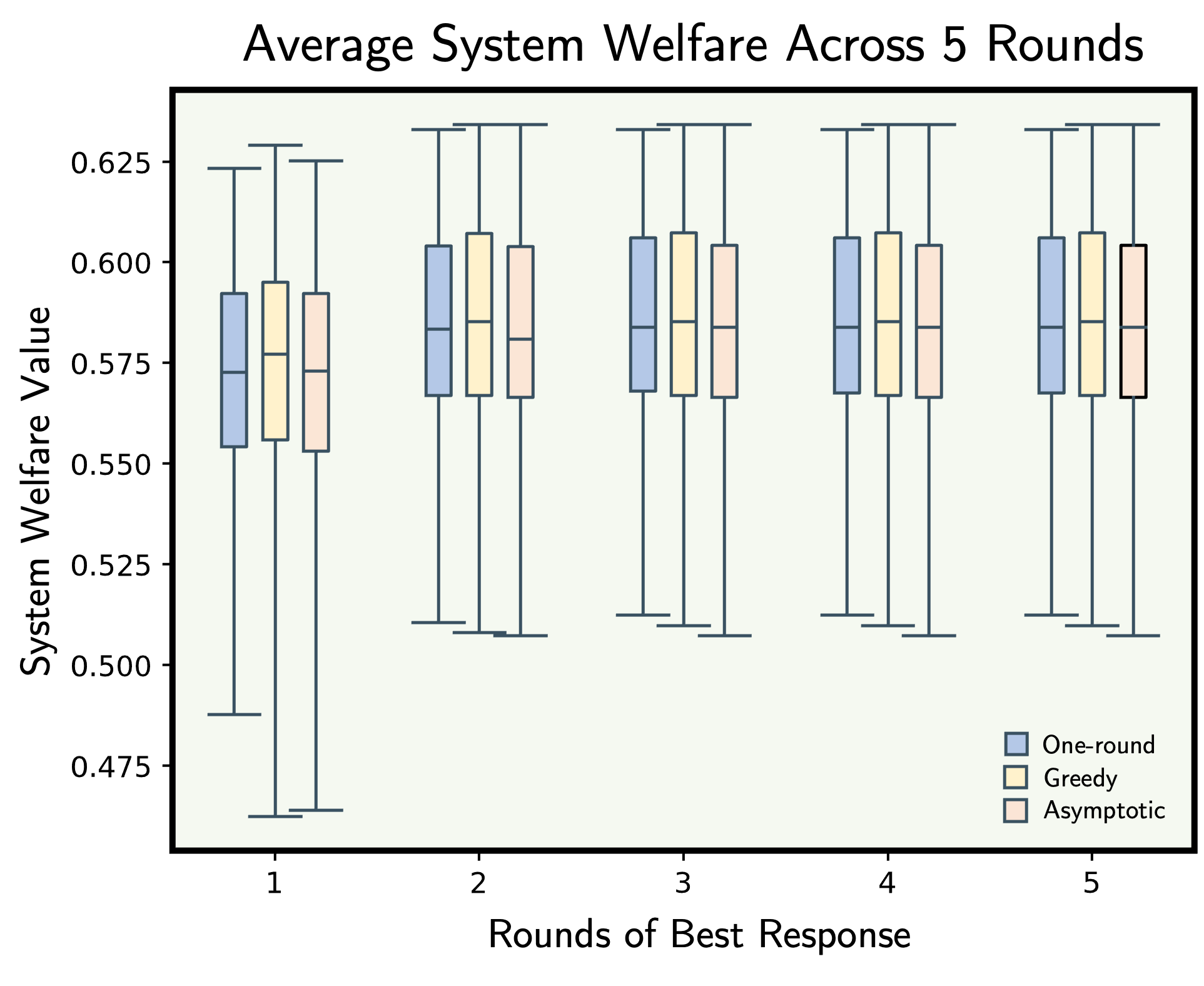}
    \caption{We plot the average rate of event detection in a randomly generated set of wireless sensor coverage problems with respect to three utility designs: the one-round optimal, the common interest, and the asymptotically optimal utility design. We see that in the short term, the one-round optimal design performs better in the worst case than the greedy and the asymptotically optimal utility designs.}
    \label{fig:numerical}
\end{figure}

\section{Conclusion}
\label{sec:conc}

This paper focuses on optimizing the performance of transient behavior in distributed resource allocation problems. Using the language of game theory, this work shifts the emphasis from  studying  the  system-level performance at equilibria, e.g., price of anarchy, to studying the transient performance as defined by a round robin best response process. Focusing on the class of submodular resource allocation problems, this work demonstrates that one can derive utility functions such that the performance after a single round of agent updates is relatively close to the performance of the best polynomial-time algorithms. However, we show that extending beyond one round of updates to a finite number of updates offers no potential gains in performance guarantees. Lastly, we characterize the Pareto frontier between the performance guarantees of the best response process with a single round and an arbitrary number of rounds. Future work may comprise of extending the results to other game models or consider average-case efficiency analysis.

\bibliographystyle{ieeetr}
\bibliography{references.bib}

\appendix
\label{sec:appendix}

We iterate through the proofs of the main theorems of the paper, as well as provide relevant technical discussion and lemmas. Relevant code is found at \cite{konda2022}. The outline of the provided technical proofs are as follows.
\begin{enumerate}
    \item[(A)] Given a welfare set $\ww$ and utility design $\fw$, a simplified linear program is formulated to characterize $\pob(\setgm; 1)$.
    \item[(B)] A proof is given for Theorem \ref{thm:optLP} to determine optimal utility rules and one-round efficiency guarantees.
    \item[(C)] A proof of Theorem \ref{thm:oneroundC} is provided, where the results of Theorem \ref{thm:optLP} are refined for a class of welfare rules with a given curvature.
    \item[(D)] A proof of Proposition \ref{prop:effMCC} is detailed to give a characterization of $\pob(\mathcal{G}_{\ww, \rm{CI}}; 1)$ for the common interest design.
    \item[(E)] A proof of Theorem \ref{thm:kroundC} is outlined, where an upper bound for $\pob^*(\ww; \k)$ and $\pob(\mathcal{G}_{\ww, \rm{CI}}; \k)$ is shown.
    \item[(F)] A proof of Theorem \ref{thm:submodtrade} is provided, where we describe the resulting asymptotic and transient trade-offs for the total set of submodular welfare rules.
    \item[(G)] A proof of Theorem \ref{thm:poapobtradeoff} is given, where the Pareto-frontier for $\poa(\mathcal{G}_{\wsc, \f})$ and $\pob(\mathcal{G}_{\wsc, \f}; 1)$ is characterized for set-covering games.
    \item[(H)] Given a welfare set $\ww$ that is now instead super-modular, a characterization of the optimal one-round, $\k$-round, and asymptotic efficiency guarantees is given.
\end{enumerate}

\noindent \emph{Notation.} Given a set $\S$, $\card{\S}$ represents its cardinality and $\ind{\S}$ describes the corresponding indicator function. ($\ind{\S}(e) = 1$ if $e \in \S$, $0$ otherwise). We denote the index of the $j$'th  component of a vector $\mathbf{v}$ with $\mathbf{v}_j$ or $\mathbf{v}(j)$ interchangeably. We use $\ones$ to denote a vector of all ones and $\zeros$ to denote a vector of all zeros. We sometimes use the denotation $\w(0) = \f(0) = 0$ for any welfare or utility rule.

\subsection{Linear Program Formulation of One Round Walk}

We first give a linear program that computes the efficiency $\pob(\setgmn; 1)$ that is based on a search for a worst case game construction $\G \in \setgmn$ that achieves the worst efficiency ratio for one-round. Here, $\setgmn$ denotes the set of games with a fixed $n$ number of agents, set of welfare rules $\ww$ and utility design $\fw$. A comparable primal-dual approach was also explored in \cite{paccagnan2018utility} and \cite{bilo2018unifying} for different settings. We note that it is possible to extend the linear program for $\pob(\setgmn, \k)$ for $\k > 1$ rounds, but the program becomes intractable in general.

First, we apply a key observation that for a game $\G$, truncating the action set of each agent $i$ to $\ac_i = \{\emp_i, \abr_i, \aopt_i \}$ does not affect the efficiency metric $\pob(\G; 1)$. Here, $\emp_i$ is the null action that does not select any resources, $\abr_i$ is the action that agent $i$ plays after the one-round walk is completed with $\abr \in \sol(1)$, and $\aopt_i$ is the action that agent $i$ plays in a joint action that optimizes the welfare $\aopt = \arg \max_{a \in \ac} \W(a)$.\footnote{Note that $\abr_i$ and $\aopt_i$ may be the same action, but using separate denotations does not affect the game structure. Additionally, if $\abr$ is not unique, then the one that performs the worst with respect to $\W$ is selected.} Therefore, we can restrict attention to the class of games $\mathcal{G}_{\ww, \fw}^{n, 3}$ where agents only have these three actions available without loss of generality. Furthermore, scaling $\W$ uniformly does not affect the ratio $\pob(\G; 1) = \frac{\W(\abr)}{\W(\aopt)}$, and we can assume that $\W(\abr) = 1$ without loss of generality. So we aim to find a game that maximizes the optimal welfare $\W(\aopt)$ to provide the lowest ratio. Consolidating the previous observations results in the following optimization problem
\begin{align}
&\pob(\setgmn; 1)^{-1} = \label{eq:LPinf} \\
&\max_{\G \in \mathcal{G}_{\ww, \fw}^{n, 3}} \quad \W(\aopt) \quad \textrm{subject to:} \label{eq:LPinfmax}\\
& \W(\abr) = 1, \label{eq:LPinfequal}\\
  &\U(\abr_{j \leq i}, \emp_{j > i}) \geq \U(\abr_{j < i}, \aopt_i, \emp_{j > i}) \quad \forall i \in \p, \label{eq:LPinfbr}
\end{align}

The constraint inequality in Eq. \eqref{eq:LPinfbr} maintains that the joint action $\abr$ is indeed the joint action that results after a one-round walk, where each agent $i$'s best response is $\abr_i$ (over $\aopt_i$) given that the previous $j \leq i$ agents have also played $\abr_j$. To formulate the linear program from the optimization problem in Eq. \eqref{eq:LPinf}, some necessary definitions are introduced. The possible resource allocations is enumerated by the following product set
\begin{equation*}
    \pt = \prod_{i \in \p} \{\emptyset, \{\abr_i\}, \{\aopt_i\}, \{\abr_i, \aopt_i\} \},
\end{equation*}
where each resource is classified with the set of actions that select it by each agent. Then some corresponding vectors in $\{0, 1\}^n$ can be defined.
\begin{align*}
    \babr_i &= \Big\{ 1 \text{ if } \abr_i \in p_i, \ 0 \text{ otherwise} \Big\}, \\
    \baopt_i &= \Big\{  1 \text{ if } \aopt_i \in p_i, \ 0 \text{ otherwise} \Big\},
\end{align*}
where $p \in \pt$ describes a resource type. We define the norm of $\babr$ to be $\card{\babr} = \sum_{i \in \p} \babr_i$ and denote the number of nonzero elements before index $i$ as $\card{\babr}_{< i} = \sum_{1 \leq j < i} \babr_j$ (similarly for $\card{\baopt} = \sum_{i \in \p} \baopt_i$). With this, we describe the linear program in the following lemma.

\begin{lemma}
\label{lem:LPprim}
Consider the welfare set $\ww = \{\w^1, \dots, \w^m\}$ with $\w^{\ell}(1) = 1$, and the corresponding utility design $\fw(\w^{\ell}) = \f^{\ell}$ with $\f^{\ell}(1) = 1$ for all $\ell$. For $n$ agents, the one-round walk efficiency is $\pob(\setgmn; 1) = \min_{1 \leq \ell \leq m} \frac{1}{\pbdual^{\ell}}$ \footnote{Here, we assume that $\ww$ is a finite set for ease of exposition, but it is straightforward to extend the efficiency result to an uncountable set.}, where $\pbdual^{\ell} \in \R$ is the solution to
\begin{align}
    \pbdual^{\ell} = &\min_{\{\lambda_i\}_{i \in \p}, \pbdual} \quad \pbdual  \quad \textrm{subject to:} \nonumber \\
    & \pbdual \w^{\ell}(\card{\babr}) \geq \w^{\ell}(\card{\baopt}) \ +  \label{eq:LPdualconst} \\
 & \sum_{i \in \p} \lambda_i \Big[ \big(\babr_i - \baopt_i \big) \f^{\ell}(\card{\babr}_{< i} + 1) \Big] \quad \forall p \in \pt \nonumber \\
  &\lambda_i \geq 0 \quad \forall i \in \p. \nonumber
\end{align}
\end{lemma}
\begin{proof}
First we show the equivalence of the optimization program proposed in Eq. \eqref{eq:LPinf} and the primal linear program described below. We later show that the dual of the program below is exactly the linear program described in the lemma. Here, we use $\pbdual = \pob(\setgmn; 1)^{-1}$ to denote the efficiency guarantee.
\begin{align}
    &\pbdual = \max_{\{\pbvar\}_{\ell, p \in \pt}} \quad \sum_{\substack{1 \leq \ell \leq m, \\ p \in \pt}} \w^{\ell}(\card{\baopt}) \cdot \pbvar \quad \textrm{subject to:} \label{eq:LPprimalmax} \\
&\sum_{\substack{1 \leq \ell \leq m, \\ p \in \pt}} \w^{\ell}(\card{\babr}) \cdot \pbvar = 1 \label{eq:LPprimaleq} \\
  &\sum_{\substack{1 \leq \ell \leq m, \\ p \in \pt}} \Big[ \big(\babr_{i} - \baopt_{i} \big) \f^{\ell}(\card{\babr}_{< i} + 1) \Big] \cdot \pbvar \geq 0 \ \ \ \forall i \in \p \label{eq:LPprimalineq} \\
  &\pbvar \geq 0 \quad \forall p \in \pt, \ 1 \leq \ell \leq m. \label{eq:LPprimalval} 
\end{align}
Here, each decision variable $\pbvar \in \R_{\geq 0}$ is a real non-negative number. We define a vector label for each resource $r$ as $\lblr(i) = \{a_i \in \ac_i: \text{ if } r \in a_i \}$. This function describes in what actions is the resource selected by each agent $i$, with $\lblr \in \pt$. Furthermore, we denote the specific partition of the resource set with $\rrsub = \{r \in \rr: \lblr = p, \w_r = \w^{\ell}\}$. Now we show that $\W(\aopt)$ in Eq. \eqref{eq:LPinfmax} matches Eq. \eqref{eq:LPprimalmax}.
\begin{align*}
    \W(\aopt) &= \sum_{r \in \rr}{\w_r(|\aopt|_r)} \\
    &= \sum_{\substack{1 \leq \ell \leq m, \\ p \in \pt}} \  \sum_{r \in \rrsub} \w^{\ell}(|\aopt|_r) \\
    &= \sum_{\substack{1 \leq \ell \leq m, \\ p \in \pt}} \w^{\ell}(\card{\baopt}) \cdot \pbvar,
\end{align*}
where $\pbvar = \card{\rrsub} \in \N$. The first equality is from the definition of the welfare function. The second equality results from partitioning the resource set. The third equality occurs by the fact that $|\aopt|_r = \sum_{j \in \p} \ind{\aopt_j}(r) = \card{\baopt}$ if $r \in \rrsub$; additionally, the value $\w^{\ell}(\card{\baopt})$ is constant for any $r \in \rrsub$. A similar argument can be made about the welfare of the best response action $\W(\abr)$, so Eq. \eqref{eq:LPinfequal} matches Eq. \eqref{eq:LPprimaleq} as well. 

Now we show the utility constraint in Eq. \eqref{eq:LPinfbr} matches the constraint in Eq. \eqref{eq:LPprimalineq}. For conciseness, let $a^1 = (\abr_{j < i}, \abr_i, \emp_{j > i})$ and $a^2 = (\abr_{j < i}, \aopt_i, \emp_{j > i})$. The utility difference can be written as

\begin{align*}
    &\U(a^1) - \U(a^2) = \sum_{r \in \abr_i} \f_r(|a^1|_r) - \sum_{r \in \aopt_i} \f_r(|a^2|_r)\\
    &= \sum_{r \in \rr} \Big( \ind{\abr_i}(r) \f_r(|a^1|_r) - \ind{\aopt_i}(r) \f_r(|a^2|_r) \Big) \\
    &= \sum_{\substack{1 \leq \ell \leq m, \\ p \in \pt}} \  \sum_{r \in \rrsub} \Big( \ind{\abr_i}(r) \f_r(|a^1|_r) - \ind{\aopt_i}(r) \f_r(|a^2|_r) \Big) \\
    &= \sum_{\substack{1 \leq \ell \leq m, \\ p \in \pt}} \  \sum_{r \in \rrsub} \Big[ \big(\babr_{i} - \baopt_{i} \big) \f^{\ell}(\card{\babr}_{< i} + 1) \Big] \\
    &= \sum_{\substack{1 \leq \ell \leq m, \\ p \in \pt}} \Big[ \big(\babr_{i} - \baopt_{i} \big) \f^{\ell}(\card{\babr}_{< i} + 1) \Big] \pbvar.
\end{align*}

The first equality is from the definitions of the utility functions. The second and third equalities comes from rewriting the sum using indicator functions and partitioning the resource set along $\pt$. The fourth equality is a result of three facts: that $\ind{\abr_i}(r) = \babr_{i}$; that $\ind{\aopt_i}(r) = \baopt_{i}$; that $|a^1|_r = \sum_{j < i} \ind{\abr_j}(r) + 1 = \card{\babr}_{< i} + 1$ if $r \in \abr_i$ (similarly for $|a^2|_r$). The fifth equality comes from sliding out the relevant terms of the first sum.

The constraint in Eq. \eqref{eq:LPprimalval} ensures a well-defined non-degenerate game parametrization. Observe that in the primal program in Eq. \eqref{eq:LPprimalmax}, we have relaxed $\pbvar \in \N$ to $\pbvar \in \R$, where we have normalized the number of resources in each partition so that $\W(\abr) = 1$. This is done without loss of generality, since we can scale up the optimal arguments $\{\pbvar\}_{\ell, p \in \pt}$ uniformly and round to derive a corresponding valid game construction that achieves an efficiency ratio $\pob(\G; 1)$ that is arbitrarily close to the solution of the primal program.

We now verify that the dual of the program in Eq. \eqref{eq:LPprimalmax} is the one in Eq. \eqref{eq:LPdualconst}. Note that primal program in Eq. \eqref{eq:LPprimalmax} can be concisely written as
\begin{align*}
    &\max_{\eta} \quad c^T \eta \quad \textrm{subject to:} \\
\quad & A \eta = 1 \\
  & \begin{bmatrix} H \\ I_{m \cdot 4^n} \end{bmatrix} \eta \succeq 0,
\end{align*}
where $\eta$ is the vector of $\{\pbvar\}_{\ell, p \in \pt}$, $I_{m \cdot 4^n}$ corresponds to the identity matrix of dimension $m \cdot 4^n \times m \cdot 4^n$, and $c$, $A$, $H$ are the compactly written vectors in equations \eqref{eq:LPprimalmax}, \eqref{eq:LPprimaleq}, and \eqref{eq:LPprimalineq} respectively. Writing the dual linear program gives
\begin{align*}
    &\max_{\lambda \succeq 0, \  \xi \succeq 0, \ \pbdual} \quad -\pbdual \quad \textrm{subject to:} \\
    & A_\ell^T \pbdual - \begin{bmatrix} H_\ell^T,  I_{4^n} \end{bmatrix} \begin{bmatrix} \lambda \\ \xi \end{bmatrix} - c_\ell = \zeros \quad \forall 1 \leq \ell \leq m,
\end{align*}
where $\zeros$ is a vector of zeros, and $c = (c_1^T, \dots, c_m^T)^T$ associated with each $1 \leq \ell \leq m$ (likewise for $A$ and $H$). Observing that $A_\ell^T \pbdual - \begin{bmatrix} H_\ell^T,  I_{4^n} \end{bmatrix} \begin{bmatrix} \lambda \\ \xi \end{bmatrix} - c_\ell = \zeros$ is equivalently written as $A_\ell^T \pbdual - H_\ell^T \lambda - c_\ell = \xi$ and as $A_\ell^T \pbdual + H_\ell^T \lambda + c_\ell \succeq 0$ and substituting back $c_\ell$, $A_\ell$, $H_\ell$ gives the result.
\end{proof}

While we have an exact characterization of the one-round walk efficiency, we cannot use this program directly to derive efficiency bounds tractably. However, by reasoning about the tight constraints in dual program, we can arrive at a more tractable program when the number of agents is not fixed.

\begin{lemma}
\label{lem:tractableLP}
Consider the welfare set $\ww = \{\w^1, \dots, \w^m\}$ with $\w^{\ell}(1) = 1$, and the corresponding utility design $\fw(\w^{\ell}) = \f^{\ell}$ with $\f^{\ell}(1) = 1$ for all $\ell$. The one round walk efficiency is $\pob(\setgm; 1) = \min_{1 \leq \ell \leq m} \frac{1}{\pbdual^{\ell}}$, where $\pbdual^{\ell} \in \R\cup \{\infty \}$ is the solution to
\begin{align}
& \pbdual^{\ell} = \min \quad \pbdual \quad \textrm{subject to:}  \label{eq:dualtractLP} \\
& \pbdual \w^{\ell}(\ay) \geq \rm{H}^{\ell} \left( \sum_{i=1}^\ay \f^\ell(i) - \bz \min_{1 \leq i \leq \ay+1} \f^\ell(i) \right) + \w^\ell(\bz) \nonumber \\
&\text{for all } \bz, \ay \in \N \text{ s.t. } \bz \geq 0 \text{ and } \ay \geq 1 \nonumber,
\end{align}
and $\rm{H}^{\ell} = \sup_{i} \w^{\ell}(i)/i$. 
\end{lemma}

\begin{proof}
The dual program in Eq. \eqref{eq:LPdualconst} provides a solution for $\pbdual^{\ell}$ for fixed $n$ agents. We first show the solution is upper bounded by $\pbdual^{\ell} \leq \Tilde{\pbdual}^{\ell}$ for any $n$, where $\Tilde{\pbdual}^{\ell}$ is the solution to the program in Eq. \eqref{eq:dualtractLP}. 

Let $n$ be the number of agents. For a given $p \in \pt$, we denote $\ay_p = \card{\babr}$ and $\bz_p = \card{\baopt}$ for ease of notation. Additionally, to convey which indices the resource type $p$ are non-zero in and in what order, we define vectors $\iabr$ for $\abr$ and $\iaopt$ for $\aopt$. Formally, $\iabr: \{1, \dots, \ay_p\} \to \{1, \dots, n\}$ and $\iaopt: \{1, \dots, \bz_p\} \to \{1, \dots, n\}$ with 
\begin{align*}
    \iabr(j) = i \text{ if } \babr_{i} = 1 \text{ and } \card{\babr}_{\leq i} = j, \\
    \iaopt(j) = i \text{ if } \baopt_{i} = 1 \text{ and } \card{\baopt}_{\leq i} = j.
\end{align*}

Considering the dual program in Eq. \eqref{eq:LPdualconst}, we add the constraint that $\lambda_i = \rm{H}^{\ell} =  \max_{1 \leq j \leq n} \w^{\ell}(j)/j$ explicitly. Since we shrink the feasible region, the optimal solution to Eq. \eqref{eq:LPdualconst} potentially increases. We verify that the resulting feasible region is nonempty. Consider the constraints according to $p$ such that $\babr = \zeros$. The corresponding dual constraint takes the form
\begin{equation*}
    0 \geq \w^{\ell}(\bz_p) - \sum_{j = 1}^{\bz_p} \lambda_{\iaopt(j)} \f^\ell(1).
\end{equation*}
Simplifying the expression gives $\sum_{j = 1}^{\bz_p} \lambda_{\iaopt(j)} \geq \w^{\ell}(\bz_p)$, which is always satisfied if $\lambda_i = \rm{H}^{\ell}$ for all $i$. If the constraints according $p$ are such that $\babr \neq \zeros$, then $\pbdual \w^\ell(\ay)$ is present and strictly positive in the inequality \eqref{eq:LPdualconst} and $\pbdual$ can be taken as high as needed to satisfy the constraint. Therefore the feasible region is nonempty.

For any $p \in \pt$ such that $\babr \neq \zeros$, we can simplify the dual constraint in Eq. \eqref{eq:LPdualconst} to 
{\small
\begin{equation*}
    \pbdual \w^{\ell}(\ay_p) \geq \w^{\ell}(\bz_p) + \sum_{i = 1}^{\ay_p} \rm{H}^{\ell} \f^\ell(i) - \sum_{i \in \p} \rm{H}^{\ell} \baopt_{i} \f^{\ell}(\card{\babr}_{<i} + 1).
\end{equation*}
}
Furthermore, for any $p \in \pt$, we observe that $\sum_{i \in \p} \baopt_{i} \f^{\ell}(\card{\babr}_{<i} + 1) \geq \bz_p \min_{1 \leq i \leq \ay_p+1} \f^\ell(i)$. Thus, for any $p \in \pt$, we can replace the corresponding dual constraint with a more binding constraint
\begin{equation*}
    \pbdual \w^{\ell}(\ay) \geq \w^{\ell}(\bz) + \sum_{i = 1}^{\ay} \rm{H}^{\ell} \f^\ell(i) - \sum_{i \in \p} \rm{H}^{\ell} \bz \min_{1 \leq i \leq \ay+1} \f^\ell(i),
\end{equation*}
for some $0 \leq \bz \equiv \bz_p \leq n$ and $1 \leq \ay \equiv \ay_p \leq n$. Therefore, replacing the dual constraints gives an upper bound for $\pbdual^{\ell} \leq \Tilde{\pbdual}^{\ell}$. Limiting the number of agents $n \to \infty$ results in the program in Eq. \eqref{eq:dualtractLP}.

Now we show that the solution is lower bounded by $\pbdual^{\ell} \geq \Tilde{\pbdual}^{\ell}$, where $\Tilde{\pbdual}^{\ell}$ is the solution to the program in Eq. \eqref{eq:dualtractLP}. We show that when we remove dual constraints, we arrive at the program in Eq. \eqref{eq:dualtractLP}. Since the feasible region expands, the optimal solution potentially decreases. Let the set of agents be $\p = \N$ and $j^{p} = \arg \min_{1 \leq j \leq \ay_p+1} \f^\ell(j)$. We remove all the dual constraints barring the constraints that correspond to $p \in \pt$ with either (\textbf{a}) $\ay_p = 0$ and $\bz_p = \bz^* = \arg \max \w^{\ell}(j)/j$ or (\textbf{b}) $\ay_p > 0$ and $\iabr(j^{p} - 1) < \iaopt(1)$ and $\iaopt(\bz_p) < \iabr(j^{p})$. The first property refers to all resource types where $\abr$ is never selected but $\aopt$ is by $\bz^*$ agents. The second property refers to all resource types where the indices of the agents selecting $\aopt$ are between the agents with index $\iabr(j^{p}-1)$ and $\iabr(j^{p})$.

Assume property (a). Then the corresponding dual constraint in Eq. \eqref{eq:LPdualconst} can be written as
\begin{equation*}
    0 \geq \w^{\ell}(\bz^*) - \sum_{j = 1}^{\bz^*} \lambda_{\iaopt(j)} \f^\ell(1),
\end{equation*}
for any resource type $p \in \pt$ that satisfies property (a). Therefore, for any $j \in \N$, except for at most $\bz^*-1$ values, observe that $\lambda_j \geq \rm{H}^{\ell}$ must hold.

Now assume property (b). With respect to a resource type $p \in \pt$ that satisfies property (b), we observe that $\f^{\ell}(\card{\babr}_{<i} + 1) = \f^{\ell}(j^{p})$ for any agent with index $i = \iaopt(j)$ for some $j$. Therefore, under the two previous observations, we can rewrite the relaxed dual program as 

\begin{align}
&\min_{\lambda \succeq \zeros} \quad \pbdual \quad \textrm{subject to:} \label{eq:LPredcont} \\
& \pbdual w^{\ell}(\ay_p) \geq \sum_{j = 1}^{\ay_p} \lambda_{\iabr(j)} \f^{\ell}(j) - \sum_{j = 1}^{\bz_p} \lambda_{\iaopt(j)} \f^{\ell}(j^{p}) + \w^{\ell}(\bz_p) \nonumber \\
&\text{for all } p \in \pt', \nonumber \\ 
&\lambda_i \geq \rm{H}^{\ell} \ \  \text{for all $i \in \N$ but at most $\bz^*-1$ values,} \nonumber
\end{align}
where $\pt' = \{p \in \pt : p \text{ satisfies property (b)} \}$. Assuming that the optimal dual variable is $\lambda_i = \rm{H}^{\ell}$ for all $i \in \N$, observe that we recover the proposed program given in Eq. \eqref{eq:dualtractLP}. To show this claim, we confirm that the binding constraint for $\pbdual$ in Eq. \eqref{eq:LPredcont} is larger when considering a different sequence of lambdas $\lambda \neq \rm{H}^{\ell} \ones$. In other words for a given $\ay \geq 1$ and $\bz \geq 0$, we show that for the resulting dual variables,
\begin{align}
    \lamallbeta &= \max_{p \in \pt'} \Big\{ \frac{1}{\w^{\ell}(\ay_p)} \big( \sum_{j = 1}^{\ay_p} \lambda_{\iabr(j)} \f^{\ell}(j) - \sum_{j = 1}^{\bz_p} \lambda_{\iaopt(j)} \f^{\ell}(j^{p})  \big) \Big\} \nonumber \\
    &\geq \frac{\rm{H}^{\ell}}{\w^{\ell}(\ay)}\left( \sum_{j = 1}^{\ay} \f^{\ell}(j) - \sum_{j = 1}^{\bz} \f^{\ell}(j^{p}) \right) = \lamonebeta \label{eq:lamgeqone}
\end{align}

For any $\lambda \neq \rm{H}^{\ell} \ones$, consider two cases where either $\lambda$ is a divergent sequence, or it is bounded above. In the first case, since $\lambda$ must satisfy $\lambda_j \geq 0$ for all $j \in \N$, the limit $\lim_{j \to \infty} \lambda_j = \infty$. If $\f^{\ell}(j) = 0$ for all $j$, note that $\lamonebeta = 0$ for any $\ay \geq 1$ and $\bz \geq 0$. Since $\lamallbeta$ must also be greater than $0$, the inequality in Eq. \eqref{eq:lamgeqone} holds in this case. If $\f^{\ell}(J) > 0$ for some $J \in \N$, consider a constraint with $p$ such that $\ay_p > J$ and $\bz_p = 0$. For any $M > 0$, we can choose $\iabr$, such that $\lambda_{\iabr(j)} > M$ for all $1 \leq j \leq \ay_p$. Thus $\lamallbeta \geq  \frac{1}{\w(\ay_p)} \sum_{j = 1}^{\ay_p} M \f(j)$. Since $M$ is arbitrary, $\lamallbeta = \infty \geq \lamonebeta$ for any $\ay \geq 1$ and $\bz \geq 0$ as well.

In the second case, since $\lambda$ is also bounded below by $\rm{H}^{\ell}$, for all but a finite set of values, there exists a convergent sub-sequence $\lamss$ that converges to a value $V \geq \rm{H}^{\ell}$ by the Bolzano-Weierstrauss theorem. Let $\fmaxj = \max_{1 \leq j \leq \ay+1} \f^{\ell}(i)$, $\myz = \max(\ay, \bz)$, and $\varepsilon > 0$. Since $\lamss$ converges, there exists a $J \in \N$ such that for any $j \geq J$, $\card{\lamss(j) - V} \leq \frac{\varepsilon}{2 \fmaxj \myz}$. 

For a given $\ay$ and $\bz$, consider any constraint with $p \in \pt'$ such that $\ay_p = \ay$ and $\bz_p = \bz$. Additionally, $\iabr$ and $\iaopt$ can be chosen to ensure that $\card{\lambda_{\iabr(j)} - V} \leq \frac{\varepsilon}{2 \fmaxj \myz}$ and$\card{\lambda_{\iaopt(j)} - V} \leq \frac{\varepsilon}{2 \fmaxj \myz}$ for all $j$. Therefore
\begin{align*}
    \lamallbeta  &\geq \frac{1}{\w^{\ell}(\ay_p)} \big( \sum_{j = 1}^{\ay_p} \lambda_{\iabr(j)}\f^{\ell}(i) - \sum_{j = 1}^{\bz_p} \lambda_{\iaopt(j)} \f^{\ell}(j^{p}) \big) \\
    &\geq \frac{V}{\w^{\ell}(\ay)} \big( \sum_{j = 1}^{\ay} \f^{\ell}(i) - \sum_{j = 1}^{\bz} \f^{\ell}(j^{p}) \big) - \frac{\varepsilon}{2} - \frac{\varepsilon}{2} \\
    &\geq \lamonebeta - \varepsilon.
\end{align*}
Since $\varepsilon$ is arbitrary, we have that $\lamallbeta \geq \lamonebeta$ for any $\ay$ and $\bz$ and we show the claim. Therefore the proposed program is an upper bound and we have shown the equality $\pbdual^{\ell} = \Tilde{\pbdual}^{\ell}$.
\end{proof}

\subsection{Proof of Theorem \ref{thm:optLP}}
\label{subsec:proofoptLP}

Given a set of welfare rules and utility rules, Lemma \ref{lem:tractableLP} provides an exact characterization of the one-round walk efficiency through a linear program. We modify the linear program in Eq. \eqref{eq:dualtractLP} to compute the utility rules that optimize the one-round walk efficiency. If a given welfare rule $\w$ is submodular, note that $\sup_i \w(i)/i = 1$ and so $\rm{H}^{\ell} = 1$. Furthermore, if the utility rule $\f$ is assumed to be non-increasing, then $\min_{1 \leq i \leq \ay+1} \f(i) = \f(\ay+1)$. Additionally, $\w(1) - 1 \cdot \f(\ay+1) \geq \w(0) - 0 \cdot \f(\ay+1) = 0$ for any $\ay \geq 1$, so $\bz = 0$ is a nonbinding constraint. We lastly note that the values $\{\f(i)\}_{i \in \p}$ can be established as decision variables for the program in Eq. \eqref{eq:dualtractLP} to produce the linear program in Eq. \eqref{eq:submodLPoptimal}, rewritten below.

\begin{align}
    &(\pbdual^*, \f^*) \in \arg \min_{\pbdual, \{\f(i)\}_{i \in \p}} \quad \pbdual \quad \textrm{subject to:}  \label{eq:submodLPoptimal2} \\
& \pbdual \w(\ay) \geq \sum_{i=1}^\ay \f(i) - \bz \f(\ay+1) + \w(\bz) \quad \forall \ay, \bz \geq 1 \nonumber \\
 &\f(1) = 1, \nonumber
\end{align}
where $\pbdual^*$ is a tight characterization of the efficiency guarantee only if the resulting optimal utility rule $\f^*$ is non-increasing and a lower bound if not. We now verify that the optimal utility rule $\f^*$ is indeed non-increasing. First, rearranging the terms in the constraint in Eq. \eqref{eq:submodLPoptimal2} gives that for any $\ay \geq 1$,
\begin{equation}
    \label{eq:supcondopt1}
    \f^*(\ay+1) \geq \sup_{\bz \geq 1} \Big( \frac{1}{\bz} \big( \sum_{i=1}^{\ay} \f^*(i) + \w(\bz) - \pbdual^* \w(\ay) \big) \Big).
\end{equation}

We verify $\f^*(\ay+1)$ is well-defined. Note that since $\f^*$ is optimal, the efficiency bound $\pbdual^* < \infty$ is nontrivial (as $\fmc$ guarantees an efficiency guarantee greater than $1/2$ according to Subsection \ref{sub:MCC}). Then, by recursion and the fact that $\frac{\w(\bz)}{\bz} \leq 1$ for all $\bz$, there exists a solution for $\f^*(\ay+1)$ such that Eq. \eqref{eq:supcondopt1} holds with equality and the resulting value is finite for all $\ay \geq 1$. Additionally $\f^*(\ay)$ must be non-negative for all $\ay \geq 1$, since limiting $\bz \to \infty$ in Eq. \eqref{eq:supcondopt1} gives that $\f(\ay+1) \geq 0$.

Now we show that the solution $\f^*$ is non-increasing. Suppose for contradiction that for some $\ay \geq 1$, that $\f^*(\ay) < \f^*(\ay+1)$. Let $\bz_{\ay+1} \in \arg \max_{\bz \geq 1} \w(\bz) - \bz \f(\ay+1)$ be the number that achieves the maximum.

We verify that $\bz_{\ay+1}$ is well-defined. Suppose for contradiction that $\w(\bz) - \bz \f^*(\ay+1)$ is always increasing in $\bz$, so $\bz_{\ay+1}$ is not well defined. Since $\pbdual^* < \infty$, the limit $\lim_{\bz \to \infty}\w(\bz) - \bz \f^*(\ay+1)$ must converge and therefore $\f^*(\ay+1)$ must be equal to $Q = \lim_{\bz \to \infty} \Delta \w(\bz)$, where we denote $\Delta \w (\bz) = \w(\bz) - \w(\bz - 1)$ for conciseness. From the original contradiction assumption then $\f^*(y) < \f^*(\ay+1) = Q$. Then taking the constraint in Eq. \eqref{eq:submodLPoptimal2}, with $\ay-1$ and $\bz \to \infty$ gives $\pbdual \w(\ay-1) \geq \lim_{\bz \to \infty}\w(\bz) - \bz \f^*(\ay) \geq \infty$, which is a contradiction.

Now, substituting $\bz_{\ay+1}$ into Eq. \eqref{eq:supcondopt1} for $\ay$ and $\ay+1$ produces the following expressions
\begin{align*}
    \f^*(\ay + 1) &=  \frac{1}{\bz_{\ay+1}} \big( \sum_{i=1}^{\ay} \f^*(i) + \w(\bz_{\ay+1}) - \pbdual^* \w(\ay) \big) \\
     \f^*(\ay) &\geq \frac{1}{\bz_{\ay+1}} \big( \sum_{i=1}^{\ay-1} \f^*(i) + \w(\bz_{\ay+1}) - \pbdual^* \w(\ay-1) \big).
\end{align*}
Inputting these expressions into the assumption $\f^*(\ay) < \f^*(\ay+1)$ reduces to the inequality $\f(\ay) > \pbdual^* \Delta \w (\ay)$. Similarly, for some $j \geq 1$, substituting $\bz_{\ay+j}$ into Eq. \eqref{eq:submodLPoptimal2} for $\ay+j$ and $\ay+j+1$ gives
{\small
\begin{align*}
    &\f^*(\ay +j+1) \geq  \frac{1}{\bz_{\ay+j}} \big( \sum_{i=1}^{\ay+j} \f^*(i) + \w(\bz_{\ay+j}) - \pbdual^* \w(\ay+j) \big) \\
     &\f^*(\ay+j) = \frac{1}{\bz_{\ay+j}} \big( \sum_{i=1}^{\ay+j-1} \f^*(i) + \w(\bz_{\ay+j}) - \pbdual^* \w(\ay+j-1) \big).
\end{align*}
}%
\noindent Thus by substituting the second expression into first, the following inequality holds
\begin{equation}
\label{eq:recnoninc}
    \f^*(\ay + j + 1) \geq \f^*(\ay+j) + \frac{\f^*(\ay+j) - \pbdual^* \Delta \w (\ay+j)}{\bz_{\ay+j}}.
\end{equation}
We show, by induction, that the following expression holds for any $j \geq 1$, 
{\small
\begin{equation}
    \label{eq:induc}
    \frac{\f^*(\ay+j) - \pbdual^* \Delta \w (\ay+j)}{\bz_{\ay+j}} \geq \frac{\f^*(\ay+1) - \pbdual^* \Delta \w (\ay+1)}{\bz_{\ay+1}} > 0.
\end{equation}
}
\noindent The base case holds for $j=1$, since
\begin{equation*}
    \f^*(\ay+1) - \pbdual^* \Delta \w (\ay+1) > \f^*(\ay) - \pbdual^* \Delta \w (\ay) > 0.
\end{equation*}
This comes from the assumption that $\f^*(\ay+1) > \f^*(\ay)$, $\Delta \w (\ay+1) \leq \Delta \w (\ay)$ by submodularity of $\w$, and that $\f^*(\ay) - \pbdual^* \Delta \w (\ay) > 0$ from the previous argument. For the inductive case for $J \geq 2$, assume that the inequality holds for all $j < J$. Then, by applying the induction assumption to Eq. \eqref{eq:recnoninc} and subsequently to the definition of $\bz_{\ay+J}$, we have that 
\begin{align*}
    \f^*(\ay+J) > \f^*(\ay+J-1) &> \dots > \f^*(\ay+1) \\
    \bz_{\ay+J} \leq \bz_{\ay+J-1} &\leq \dots \leq \bz_{\ay+1}.
\end{align*}
Therefore the statement in Eq. \eqref{eq:induc} holds due to the aforementioned inequalities and the fact that $\Delta \w (\ay+J) \leq \Delta \w (\ay+1)$ due to submodularity of $\w$. Therefore Eq. \eqref{eq:induc} holds and we have that $\f^*(\ay + j + 1) \geq \f^*(\ay+j) + D$, where $D = \frac{\f^*(\ay+1) - \pbdual^* \Delta \w (\ay+1)}{\bz_{\ay+1}} > 0$. Following this, $\f^*(\ay + j) \geq \f^*(\ay+1) + D(j-1)$. 

Now consider the constraint in Eq. \eqref{eq:submodLPoptimal2} where $\ay \to \infty$ and $\bz = 0$. Since $\w(\ay) \leq \ay$,
\begin{equation}
    \pbdual^* \geq \lim_{\ay \to \infty} \frac{1}{\ay} \sum_{i=1}^{\ay} \f^*(i) \geq \infty,
\end{equation}
where the last inequality results from the fact that $\f^*(\ay) \sim \ay$ is of linear order by the previous argument. Since $\pbdual^*$ must be finite, contradiction ensues and the solution $\f^*$ must be non-increasing and the efficiency guarantees are tight for each linear program.

Thus, so far, we have shown the statement in Eq. \eqref{eq:effoptratio} with regards to the welfare set $\ww = \{\w^1, \dots, \w^m\}$. We lastly show that the results extend linearly to a span of welfare rules as claimed in Eq. \eqref{eq:optlinear}. Note that for the welfare set $\ww_{\rm{span}}$ spanned from $\{\w^1, \dots, \w^m\}$, the resulting optimal guarantees $\pob^*(\ww; 1) \geq \pob^*(\ww_{\rm{span}}; 1)$, since $\ww_{\rm{span}}$ is a larger set of welfare rules. We show that the utility design as in Eq. \eqref{eq:optlinear} achieves $\pob^*(\ww_{\rm{span}}; 1) = \pob^*(\ww; 1)$ and therefore is optimal with respect to $\ww_{\rm{span}}$. Consider $\w = \sum_{\ell=1}^m \alpha^\ell \w^\ell$ for any non-negative $\{\alpha^{\ell}\}_{1 \leq \ell \leq m}$. Let the corresponding utility design be $\fw_{\rm{lin}}(\w) = \sum_{\ell=1}^m \alpha^\ell \f^\ell$, where $\f^{\ell}$ is the corresponding solution to Eq. \eqref{eq:submodLPoptimal} for $\w^{\ell}$ and $\pbdual^* = \min_{1 \leq \ell \leq m} \frac{1}{\pbdual^\ell} = \pob^*(\ww; 1)$. From the characterization program in Eq. \eqref{eq:dualtractLP} with respect to $\ww_{\rm{span}}$ and $\fw_{\rm{lin}}$, the dual constraint for any $\ay \geq 1$, $\bz \geq 0$, $1 \leq \ell \leq m$, and $\{\alpha^{\ell}\}_{1 \leq \ell \leq m}$ can be rewritten as
\begin{equation*}
    \sum_{\ell=1}^m \alpha^\ell  \cdot \Big[\pbdual^* \w^{\ell}(y) - \sum_{j=1}^{\ay} \f^{\ell}(j) + \bz \f^{\ell}(\ay + 1) - \w^{\ell}(\bz) \Big] \geq 0.
\end{equation*}
This constraint will always be satisfied for any non-negative $\{\alpha^{\ell}\}_{1 \leq \ell \leq m}$, as the inner terms is non-negative by definition of $\pbdual^*$ and $\f^{\ell}$. Therefore, we have that under the linear utility design $\fw_{\rm{lin}}$, we have that $\pob(\mathcal{G}_{\ww_{\rm{span}}, \fw_{\rm{lin}}}) \geq \pob^*(\ww; 1)$ and is optimal.

\subsection{Proof of Theorem \ref{thm:oneroundC}}

Given a curvature $\cc$, let $\ww$ be the set of welfare rules that have curvature of at most $\cc$. From \cite[Lemma 2]{chandan2021tractable}, we know there exists a basis set of welfare rules, such that for any $\w \in \ww$, we can come up with a decomposition $w = \sum_{b \in \N} \alpha^{b} \wb$, with $\alpha^b = (2 \w(b) - \w(b-1) - \w(b+1))/\cc$ and 
\begin{equation}
\label{eq:bcccov}
    \wb(j) =
\begin{cases}
j, &\text{ if } 0 \leq j \leq b \\
b + (1-\cc) \cdot (j-b) &\text{ if } j > b.
\end{cases}
\end{equation}

We refer to these welfare rules as \emph{$b$-covering} welfare rules. We note that for any $b \in \N$, the welfare rule $\wb$ has a curvature of $\cc$. For each welfare rule $\wb$, we claim that the corresponding optimal utility rule from running the program in Eq. \eqref{eq:submodLPoptimal} is
\begin{equation}
\label{eq:optformbcc}
\fb(j) = \begin{cases}
(1 - \pbdual^b) (\frac{b+1}{b})^{j-1} + \pbdual^b & \text{ if } j \leq b + 1 \\
(1-\cc) \pbdual^b & \text{ if } j \geq b + 1,
\end{cases}
\end{equation}
where and $\pbdual^b = \frac{(\frac{b+1}{b})^b}{(\frac{b+1}{b})^b-\cc}$ is the resulting optimal efficiency. Taking the minimum across $b$, we have that $\min_{b \in \N} \frac{1}{\pbdual^b} = 1 - \cc/2$ for $b = 1$. Therefore, using Theorem \ref{thm:optLP}, the optimal efficiency guarantee is $\pob^*(\ww; 1) = 1 - \cc/2$.

Now we verify that $\fb$ and $\pbdual^{b}$ are indeed the optimal solutions. We first remove all constraints in Eq. \eqref{eq:submodLPoptimal} apart from the ones that satisfy $\bz = b$ for any $\ay \geq 1$. This results in a lower bound for $\pbdual^{b}$ that we claim later to be tight. 

Rearranging the terms in the constraint in Eq. \eqref{eq:submodLPoptimal2} gives that for any $\ay \geq 1$, the optimal solution satisfies
\begin{equation}
    \label{eq:supcondopt}
    \f^*(\ay+1) = \sup_{\bz \geq 1} \Big( \frac{1}{\bz} \big( \sum_{i=1}^{\ay} \f^*(i) + \w(\bz) - \pbdual^* \w(\ay) \big) \Big).
\end{equation}
Substituting in for $\w$ and the binding constraint $\bz = b$, the recursive equation for $\fb$ is then
\begin{align*}
    \fb(1) &= 1 \\
    \fb(j+1) &= \frac{1}{b} \sum_{i=1}^{j} \fb(i) + 1 - \frac{1}{b} \pbdual^* \wb(j),
\end{align*}
\noindent for some optimal $\pbdual^* \geq 1$. To solve for the closed form expression for $\fb$, a corresponding linear, time-invariant, discrete time system is constructed as follows.
\begin{align*}
    x_1(t+1) &= x_1(t) + x_2(t) \\
    x_2(t+1) &= \frac{1}{b}(x_1(t) + x_2(t)) + s(t) \\
    s(t) &= 1 - \frac{1}{b} \pbdual^* \wb(t).
\end{align*}
For the initial condition $(x_1(1), x_2(1)) = (0, 1)$, the corresponding solution $x_2(t) \equiv \fb(j)$. Then using the state transition matrix, we can solve for the explicit solution for $x_2(t)$ as
\begin{align}
    x_2(1) =& 1 \label{eq:recurvoptbc}\\
    x_2(t) =& \frac{1}{b} B^{t-2} + \sum_{\tau = 1}^{t-2}\frac{1}{b}B^{t-2-\tau} (1 - \pbdual^*\wb(\tau)) \nonumber \\
    & + (1 -  \pbdual^* \wb(t-1)) \quad  t > 1, \nonumber
\end{align}
where $B = \frac{b+1}{b}$. Simplifying the expression for $x_2(t)$ for $t-1 > b$ and substituting $\wb(t) = (1-\cc)t + \cc \min(t, b)$ results in the following
\begin{align*}
    x_2(t) =& \frac{1}{b} B^{t-2} \Big( 1 + \sum_{\tau = 1}^{b}B^{-\tau} (1 - \pbdual^*\tau) \\
    &+ \sum_{\tau = b+1}^{t-2} B^{-\tau} (1 - \pbdual^*((1 - \cc)\tau + \cc b)) \Big) \\
    &+ (1 - \pbdual^*(t - 1 - \cc(t-1) + \cc b)). 
\end{align*}
Now we can use the series identities $\sum_{j=1}^{d} p^j = \frac{p - p^{d-1}}{1 - p}$ and $\sum_{j=1}^{d} j p^j = \frac{p - (d+1) p^{d+1} + d p^{d+2}}{(1 - p)^2}$ and simplify the terms to 
\begin{align*}
    x_2(t) = B^{t-2}(\pbdual^* (\cc B^{1-b}- B) + B) + (1-\cc) \pbdual^*.
\end{align*}
Thus, the above expression is the closed form solution for $\fb$ when $j-1 > b$. We have already shown that the optimal utility rule $\fb$ must be non-increasing in the proof of Theorem \ref{thm:optLP}. This is only possible when $\pbdual^* \geq \frac{B^b}{B^b-\cc}$. Therefore the optimal solution must be $\pbdual^* = \pbdual^b = \frac{B^b}{B^b-\cc}$. Substituting for $\pbdual^*$ in the expression in Eq. \eqref{eq:recurvoptbc} and simplifying results in the closed form expression in Eq. \eqref{eq:optformbcc} for $\fb$. It can be seen that $\fb$ defined in Eq. \eqref{eq:optformbcc} is indeed non-increasing. We lastly verify that the binding constraint for $\fb$ is indeed when $\bz=b$ for any $\ay \geq 1$ and so $\pbdual^b$ is tight. In Eq. \eqref{eq:submodLPoptimal}, we examine the terms $\wb(\bz) - \bz \fb(\ay+1)$ for any $\ay \geq 1$. Note that $1 = \wb(\bz) - \wb(\bz-1) \geq \fb(\ay+1)$ when $\bz \leq b$ and $(1-\cc) = \wb(\bz) - \wb(\bz-1) \leq \fb(\ay+1)$ when $\bz \geq b$ for any $\ay$. Thus the maximum $\max_{\bz} \wb(\bz) - \bz \fb(\ay+1)$ occurs when $\bz = b$, and we have shown the claim.

\subsection{Proof of Proposition \ref{prop:effMCC}}
\label{sub:MCC}

Given a curvature $\cc$, let $\ww$ be the set of welfare rules that have curvature of at most $\cc$. While the common interest utility, with $\U(a) = \W(a)$ for all $a$ and $i$, is not directly implementable as a utility design, by removing the nonstrategic portion of the utility as $\U(a_i, a_{-i}) = \W(a_i, a_{-i}) - \W(\emp_i, a_{-i})$, we can arrive at the \emph{marginal contribution} (MC) utility. This utility is strategically equivalent to the common interest utility and results in an equivalent efficiency guarantee. However, the MC utility can be written as the utility design $\fw(w) = \fmc$ with $\fmc(j) = \w(j) - \w(j-1)$. Now we can use the linear program in Lemma \ref{lem:tractableLP} to characterize the efficiency guarantee of the MC utility, giving the efficiency guarantee of the common interest utility as well.

Consider any $\w \in \ww$. Since $\w$ is submodular, the utility rule $\fmc$ must be decreasing and $\sup_{j \in \N} \w(j)/j = 1$. Therefore the corresponding dual constraints in Eq. \eqref{eq:dualtractLP} can be rewritten as
\begin{equation}
\label{eq:dualMCcons}
    \pbdual \w(\ay) \geq \sum_{i=1}^\ay \fmc(i) - \bz \fmc(\ay+1) + \w(\bz),
\end{equation}
for any $\ay \geq 1$ and $\bz \geq 0$. We claim the binding constraint is when $\bz = \ay$. Fixing $\ay$, the only terms that depend on $\bz$ is $-\bz \fmc(\ay + 1) + \w(\bz)$. Examining the difference between terms from $\bz+1$ against $\bz$ gives
\begin{align*}
    &\w(\bz+1) - (\bz+1) \fmc(\ay+1) - \w(\bz) + \bz \fmc(\ay+1) \\
    &= \w(\bz+1) - \w(\bz) - \fmc(\ay+1) \\
    &= \fmc(\bz+1) - \fmc(\ay+1).
\end{align*}
Since $\fmc$ is non-increasing, note that $\fmc(\bz+1) - \fmc(\ay+1)$ is greater than $0$ if $\bz \leq \ay$ and less than $0$ if $\bz \geq \ay$. Therefore the tightest constraint is when $\bz = \ay$. Now we simplify the solution for $\pbdual^*$ in Eq. \eqref{eq:dualtractLP} under the assumption that $\bz = \ay$ as
\begin{align*}
    \pbdual^* &= \max_{\ay \geq 1} \Big\{ \frac{1}{\w(\ay)} \big( \sum_{j = 1}^{\ay} \fmc(j) - \ay \fmc(\ay+1) + \w(\ay) \big) \Big\} \\
    &= \max_{\ay \geq 1} \Big\{2 - \frac{\ay}{\w(\ay)} \fmc(\ay+1) \Big\},
\end{align*}
in which we have used the identity $\sum_{j = 1}^{\ay} \fmc(j) = \sum_{j = 1}^{\ay} \w(j) - \w(j-1) = \w(\ay)$. Since $\w$ is submodular, $\frac{j}{\w(j)} \geq 1$ for any $j \in \N$, and because $\w$ has at most curvature of $\cc$, $\fmc(j) \geq 1 - \cc$ for any $j \in \N$ as well. Therefore, the solution is upper bounded by $\pbdual^* \leq 1 + \cc$ and since $\w$ was chosen arbitrarily from $\ww$, the resulting efficiency guarantee is $\pob(\mathcal{G}_{\ww, \rm{CI}}; 1) = 1/\pbdual^* \geq (1 + \cc)^{-1}$. This efficiency guarantee is actually tight if we consider the $b$-covering welfare rule $\wb$ with curvature $\cc$, as in Eq. \eqref{eq:bcccov}. Observe that under the $b$-covering welfare, the maximum is $\max_{\ay \geq 1} \frac{\ay}{\wb(\ay)} \fmc(\ay+1) = 1 - \cc$ at $\ay = b$. Therefore, $\pob(\mathcal{G}_{\ww, \rm{CI}}; 1) = (1 + \cc)^{-1}$ with equality and we show the claim.

\subsection{Proof of Theorem \ref{thm:kroundC}}

In this section, we first provide upper bounds on the efficiency metric $\pob^*(\ww; \k)$. To do this, we construct a game $\G$ such that for any utility design $\fw$, rounds $\k \geq 1$, and curvature $\cc$, we have that $\pob(\setgm; \k) \leq  \pob(\G; \k) \leq 1 - \cc/2$. Let $\cc$ be the curvature and consider the $b$-covering rule $\wb$ with $b = 1$ as in Eq. \eqref{eq:bcccov} with $\wb(2) = 2 - \cc$. Additionally, let $\f = \fw(\wb)$ be the corresponding utility rule for a given utility design. A two-agent game $\G$ is constructed as follows. Let the resource set be $\rr = \rr_1 \cup \rr_2 \cup \rr_3$, where $\rr_j$ is a set of resources such that the ratio of resources satisfies $\card{\rr_1} = \card{\rr_2} = \f(2) \cdot \card{\rr_3}$. If $\f(2)$ is not a whole number, we can scale up $\card{\rr_j}$ uniformly and round $\f(2) \cdot \card{\rr_3}$ to get arbitrarily close to the given ratio. Let $x = \card{\rr_1}$. The action sets for the game construction the agents will be determined by $\f$ according to the following three cases: (\textbf{a}) $0 \leq \f(2) \leq (1-\cc)$, (\textbf{b}) $(1-\cc) \leq \f(2) \leq 1$, and (\textbf{c}) $\f(2) \geq 1$.

For case (a), Agent $1$'s actions are $\ac_1 = \{\emp_1, a_1^1 = \rr_1, a_1^2=\rr_2\}$. Agent $2$'s actions are $\ac_2 = \{\emp_2, a_2^1 = \rr_3, a_2^2 = \rr_1\}$. The optimal allocation is $\aopt = \{a_1^2, a_2^2\}$ resulting in a welfare of $2x$. An allocation that can occur after a one round walk is $\abr = \{a_1^1, a_2^1\}$ resulting in a welfare of $(1 + \f(2))x$. Therefore, $\pob(\G; 1) \leq \frac{(1 + \f(2))x}{2x} \leq 1 - \frac{\cc}{2}$ by assumption of $\f \leq 1-\cc$. Additionally, observe that $\abr$ is a Nash equilibrium and therefore is still the resulting allocation after any number of additional rounds $\k \geq 1$. Therefore $\pob(\setgm; \k) \leq \pob(\G; \k) \leq 1 - \frac{\cc}{2}$ for this case of utility design.

For case (b), Agent $1$'s actions are $\ac_1 = \{\emp_1, a_1^1 = \rr_1, a_1^2=\rr_2\}$. Agent $2$'s actions are $\ac_2 = \{\emp_2, a_2^1 = \rr_3, a_2^2 = \rr_1\}$. The optimal allocation is $\aopt = \{a_1^2, a_2^2\}$ resulting in a welfare of $2x$. An allocation that can occur after a one-round walk is $\abr = \{a_1^1, a_2^2\}$ resulting in a welfare of $\wb(2) \cdot x$. Therefore, $\pob(\G; 1) \leq \frac{\wb(2) \cdot x}{2 x} = 1 - \frac{\cc}{2}$. For $\k \geq 2$, there is a best response path that leads to the end state $\abr$. This is achieved by reaching $a' = \{a_1^1, a_2^1\}$ in the first round. As $a'$ is a Nash action, the best response process can remains at $a'$ for $\k-1$ rounds and in the last round, switch to $\abr$. Therefore $\pob(\setgm; \k) \leq \pob(\G; \k) \leq 1 - \frac{\cc}{2}$ for this case.

For case (c), Agent $1$'s actions are $\ac_1 = \{\emp_1, a_1^1 = \rr_1, a_1^2 = \rr_2 \}$. Agent $2$'s actions are $\ac_2 = \{\emp_2, a_2^1 = \rr_1, a_2^2 = \rr_3 \}$. The optimal allocation is $\aopt = \{a_1^2, a_2^2 \}$ resulting in a welfare of $(1 + \f(2))x$. An allocation that can occur after a one round walk is $\abr = \{a_1^1, a_2^1 \}$ resulting in a welfare of $\wb(2) \cdot x$. Therefore, $\pob(\G; 1) = \frac{\wb(2) \cdot x}{(1 + \f(2)) x} \leq 1 - \frac{\cc}{2}$ by assumption of $\f(2) > 1$. Additionally, observe that $\abr$ is a Nash equilibrium and therefore is still the resulting allocation after any number of additional rounds. Therefore $\pob(\setgm; \k) \leq \pob(\G; \k) \leq 1 - \frac{\cc}{2}$ for this case.

Since $\f = \fw(\wb)$ was chosen arbitrarily, we have that the upper bound holds for any utility design and we have shown that $\pob^*(\ww; \k) \leq 1 - \cc/2$. Furthermore, based on our game construction, the efficiency bounds hold even when we relax the class of best response dynamics that we consider. Since the game construction comprises of only two agents, allowing agents to best respond multiple times during a round or best respond out of order of round-robin does not improve the efficiency guarantees that result from the given game $\G$.

Now we show that the upper bound $\pob(\mathcal{G}_{\ww, \rm{CI}}; \k) \leq (1 + \cc)^{-1}$ as stated in Eq. \eqref{eq:klessCCI}. As before, a game $\G$ is constructed such that under the common interest design $\rm{CI}$, $\k \geq 1$, and curvature $\cc$, we have that $\pob(\mathcal{G}_{\ww, \rm{CI}}; \k) \leq  \pob(\G; \k) \leq (1 + \cc)^{-1}$. Let $\G$ have $n$ players with a resource set $\rr = \rr^{\rm{opt}} \cup \rr^{\rm{both}} \cup \{r^n\}$ with $\card{\rr^{\rm{opt}}} = n$ and $\card{\rr^{\rm{both}}} = n -1$. Each agent $i$ has three actions in its action set $\ac_i = \{\emp_i, \abr_i, \aopt_i\}$. The resources are selected by the agents in the following manner: each resource $r_j^{\rm{opt}} \in \rr^{\rm{opt}}$ is selected by agent $j$ in action $\aopt_j \ni r_j^{\rm{opt}}$ for all $1 \leq j \leq n$; each resource $r_j^{\rm{both}} \in \rr^{\rm{both}}$ is selected by agent $j+1$ in action $\aopt_{j+1} \ni r_j^{\rm{both}}$ and by agent $j$ in action $\abr_{j} \ni r_j^{\rm{both}}$ for all $1 \leq j \leq n-1$; agent $n$ selects the resource in action $\abr_n \in r^n$. Given a curvature $\cc$, consider two $b$-covering welfare rules $\wb, \wb_2 \in \ww$ with curvature $\cc$ such that $\wb(1) = 1$ and $\wb(2) = 1 - \cc$, and $\wb_2(1) = \cc$ and $\wb_2(2) = \cc(1 - \cc)$. For any $r \in \rr^{\rm{both}} \cup \{r^n\}$, let the corresponding welfare rule be $\w_r = \wb$ and for any $r \in \rr^{\rm{opt}}$, let the corresponding welfare rule be $\w_r = \wb_2$. Under this game construction it can be seen that under $\abr$, each resource $r \in \rr^{\rm{both}} \cup \{r^n\}$ is selected by exactly one agent, resulting in a welfare of $\W(\abr) = n$; also, under $\aopt$, each resource $r \in \rr^{\rm{both}} \cup \rr^{\rm{opt}}$ is selected by exactly one agent, resulting in a welfare of $\W(\abr) = (n-1)(1+\cc) + \cc$. Assuming that $\abr$ is the joint action that results after $\k$ rounds, we have that $\pob(\G, \k) \leq \frac{n}{(n-1)(1+\cc) + \cc}$. Limiting the number of agents $n \to \infty$ to infinity gives the result. To verify that $\abr$ can result after $\k$ rounds, observe that for agent $1$ selecting $\abr_1$ over $\aopt_1$ results in a higher system welfare. After that, agents $2$ through $n-1$ are indifferent between $\abr_j$ and $\aopt_j$ given that the previous $i < j$ players have selected $\abr_i$. Therefore, $\abr$ is the resulting allocation after one round. Additionally, $\abr$ can be seen to be Nash equilibrium for the common interest utility, so after any number of rounds $\k$, $\abr$ is still a possible joint allocation that can result.

\subsection{Proof of Theorem \ref{thm:submodtrade}}

We show the trade-offs in Theorem \ref{thm:submodtrade} that result from considering utility designs that maximize the one-round walk efficiency versus the price of anarchy. The fact that $\poa(\mathcal{G}_{\ww, \fw_{\poa}}) = 1 - \frac{1}{e}$ comes from \cite[Theorem 1]{chandan2021tractable} and $\pob(\mathcal{G}_{\ww, \fw_1^*}; 1) = \frac{1}{2}$ comes from setting $\cc = 1$ in Theorem \ref{thm:oneroundC}. We show that $\pob(\mathcal{G}_{\ww, \fw_{\poa}}; 1) = 0$ in Lemma \ref{lem:bCpobbad}. From Lemma \ref{lem:poblesspoa}, we have that $\poa(\mathcal{G}_{\ww, \fw_1^*}) \geq \pob(\mathcal{G}_{\ww, \fw_1^*}; 1) = \frac{1}{2}$, since $\fw_1^*$ must be a non-increasing utility design, as shown in Section \ref{subsec:proofoptLP}. To show that this lower bound is tight, consider the set covering welfare $\wsc$. As seen in Theorem \ref{thm:poapobtradeoff}, the price of anarchy guarantee is $\frac{1}{2}$, and so $\poa(\mathcal{G}_{\ww, \fw_1^*}) \leq \poa(\mathcal{G}_{\wsc, \fw_1^*(\wsc)}) = \frac{1}{2}$ as well. Now we outline Lemma \ref{lem:bCpobbad} and Lemma \ref{lem:poblesspoa} below.

\begin{lemma}
\label{lem:bCpobbad}
Suppose that $\ww$ is the set of all possible submodular welfare rules and consider the utility design $\fw_{\poa}$. Then $\pob(\mathcal{G}_{\ww, \fw_{\poa}}; 1) = 0$.
\end{lemma}
\begin{proof}
We construct a game $\G \in \mathcal{G}_{\ww, \fw_{\poa}}$ to upper bound the one-round walk efficiency such that $\pob(\G; 1) = 0$. Since $\pob(\mathcal{G}_{\ww, \fw_{\poa}}; 1)$ is defined to be greater than $0$, we have equality. Consider a game with $n$ players as follows. We partition the resource set as $\rr = \bigcup_{1 \leq j \leq n+1}\rr_j$. Every resource $r \in \rr$ is endowed the local welfare rule $\w_r = \wb$ as the $b$-covering welfare rule with curvature of $\cc = 1$ for some fixed $b \geq 1$, as defined in Eq. \eqref{eq:bcccov}. The corresponding utility rule is $\fa = \fw_{\poa}(\wb)$ is the following recursive expression from \cite[Lemma 1]{chandan2021tractable},
\begin{align*}
    \fa(1) &= 1 \\
    \fa(j+1) &= \frac{1}{b}[j \fa(j) - \rho^b \min \{j, b\}] + 1,
\end{align*}
with $\rho^b = (1 - \frac{b^b e^{-b}}{b!})^{-1}$. The number of resources in each set is $\card{\rr_1} = v$ and $\card{\rr_{j+1}} \sim v \cdot \fa(j)$ for $1 \leq j \leq n$ and for some $v \geq 0$. If $\fa(j)$ is not a whole number, we can scale $v$ up and round to get arbitrarily close to the correct ratio of resources. Agent $i$ selects $\rr_1 = \abr_i$ and $\rr_{i+1} = \aopt_i$ in each of its actions. It can be verified that $\abr$ is a joint action that can result after a one round walk. Therefore, the efficiency is upper bounded by
\begin{equation*}
\pob(\G; 1) \leq \frac{\W(\abr)}{\W(\aopt)} = \frac{vb}{v \sum_{1 \leq i \leq n}\fa(i)}.
\end{equation*}

Now we show that as we increase $n$, the series $\sum_{1 \leq i \leq n}\fa(i)$ diverges, and the efficiency can get arbitrarily bad as the number of agents increase.
To construct the closed form expression of $\fa(j)$, we construct the following LTV state space system with $\fa(j) := x(t)$
\begin{align*}
    x(t+1) &= A(t) x(t) + s(t) \quad \quad  A(t) = \frac{t}{b} \\
    s(t) &= 1 - \frac{\rho^b}{b} \min(t, b)
\end{align*}
Solving for the solution $x(t)$ using the state transition matrix with the initial condition $x(1) = 1$ results in the following expression
\begin{align*}
    x(t) &= \prod_{\tau = 1}^{t}\frac{\tau}{b} + \sum_{T=1}^{t-1} \big[ \big(1 - \frac{\rho^b}{b} \min(t, b) \big) \prod_{\tau = T+1}^{t-1} \frac{\tau}{b} \big] \\
    &= \frac{t!}{b^t} \bigg(1 + \sum_{T=1}^{t} \frac{b^T}{T!}\big(1 - \frac{\rho^b}{b} \min(t, b) \big) \bigg)
\end{align*}

If $t \geq b$, then
\begin{align*}
    x(t) &= \frac{t!}{b^t} \bigg(1 - (e^b - 1)(\rho - 1) + \sum_{T=t+1}^{\infty} \frac{b^T}{T!}\big(\rho^b - 1 \big) + \\
    &\sum_{T=1}^{b} \frac{b^T}{T!}\frac{\rho^b (b-T)}{b}\bigg) \\
    &=  \frac{t!}{b^t} \sum_{T=t+1}^{\infty} \frac{b^T}{T!}\big(\rho^b - 1 \big) \\
    &\geq (\rho^b - 1) \frac{b}{t + 1} \\
    &\sim \mathrm{O}(\frac{1}{t})
\end{align*}
The first equality results from splitting the summation and the second equality will be shown later. Since $x(t)$ is on the order of $\frac{1}{t}$, the series $\sum_{i=1}^{N} \fa(i)$ diverges and the claim is shown. Now we verify the equality
\begin{align*}
    \sum_{T=1}^{b} \frac{b^T}{T!}\frac{\rho^b (b-T)}{b} &= (e^b - 1)(\rho^b - 1) - 1 \\ 
    \sum_{T=1}^{b} \frac{b^T(b-T)}{b T!} &= \frac{1}{\rho^b} \big(e^b \rho^b - e^b - \rho^b \big) \\
    \sum_{T=1}^{b}\frac{b^{T}}{T!} - \sum_{T=1}^{b} \frac{b^{T-1}}{(T-1)!} &= \big(e^b - 1 - e^b (1 - \frac{b^b e^{-b}}{b!}) \big) \\
    \frac{b^b}{b!} - 1 &=  \frac{b^b}{b!} - 1
\end{align*}
The last equality results from recognizing the terms on the left hand side as a telescoping sum.
\end{proof}

\begin{lemma}
\label{lem:poblesspoa}
Let $\ww = \{\w^1, \dots, \w^m\}$ be a set of welfare rules and $\fw$ be a utility design such that $\f^{\ell} = \fw(\w^\ell)$ is non-increasing for any $1 \leq \ell \leq m$. Then $\pob(\setgm; \k) \leq \poa(\setgm)$ for any $\k \geq 1$.
\end{lemma}
\begin{proof}
We show this claim by a game construction, where a Nash equilibrium with the efficiency arbitrarily close to $\poa(\setgm)$ is reachable by a one-round walk. Let $\varepsilon_1 > 0$. Note that $\poa(\setgmn)$ is non-increasing in $n$ and lower bounded by $0$. Therefore $\poa(\setgmn)$ is a convergent sequence in $n$ and for any $\varepsilon_1$, there exists an $N_1 \in \N$ such that $\poa(\mathcal{G}_{\ww, \fw}^{N_1}) - \poa(\setgm) \leq \varepsilon_1$.

Generalizing \cite[Theorem 2]{paccagnan2018utility} to a set of welfare rules provides a characterization of the price of anarchy as $\poa(\mathcal{G}_{\ww, \fw}^{N_1}) = \min_{1 \leq \ell \leq m} \frac{1}{Q^{\ell}}$ with
\begin{align}
    &Q^{\ell} = \max_{\paval} \sum_{\ay, \xx, \bz} \w^{\ell}(\bz + \xx) \paval \label{eq:poaLP} \\
    \text{ s.t. } &\sum_{\ay, \xx, \bz} [\ay \f^{\ell}(\ay + \xx) - \bz \f^{\ell}(\ay+\xx+1)] \paval \geq 0 \nonumber \\
    &\sum_{\ay, \xx, \bz} \w^{\ell}(\ay + \xx) \paval = 1 \nonumber \\
    &\paval \geq 0, \nonumber
\end{align}
where $\ay, \xx, \bz, \in \N$ with $1 \leq \ay + \xx+ \bz \leq N_1$. For the $\ell^* = \arg \min_{1 \leq \ell \leq m} \frac{1}{Q^{\ell}}$ that achieves the minimum, we refer to $\w \equiv \w^{\ell^*}$, $\f \equiv \f^{\ell^*}$ for ease of notation and refer to $\paop$ to denote the corresponding optimal variables for $\paval$ of the linear program. We construct a matching game $\G$ as follows. Let $N_2 > N_1$ be the number of agents in the game and $D = N_2 + \ay + \xx - 1$. For each $\ay$, $\xx$, $\bz$ pair and $1 \leq k \leq D$, we construct a set of resources $\rraxb$ with $\card{\rraxb} = \paop / D$. Each agent $i$ has three actions in its action set $\ac_i = \{\emp_i, \ne_i, \aopt_i\}$. Each agent $i$ selects $\{\rraxb\}_{i \leq k \leq \ay + \xx + i - 1}$ in $\ne_i$ for each pair $\ay, \bz, \xx$. If $\ay + \bz + \xx \leq i \leq N_2$, agent $i$ selects $\{\rraxb\}_{i - \bz \leq k \leq \xx + i - 1}$ in $\aopt_i$ for each pair $\ay, \bz, \xx$. Otherwise for $1 \leq i \leq \ay + \bz + \xx - 1$, $\aopt_i = \emp_i$ and agent $i$ doesn't select any resources in $\aopt_i$.

We first confirm that the action $\ne$ is indeed a Nash equilibrium. Showing this for the first $\ay + \xx + \bz - 1$ agents is trivial, since no resources are selected in $\aopt_i$. For the rest of the agents, the utility difference of a unilateral deviation to $\aopt_i$ from $\ne_i$ is 
\begin{align*}
    &\U(\ne) - \U(\aopt_i, \ne_{-i}) \\
    &\geq \sum_{r \in \ne_i} \f_r(|\ne|_r) - \sum_{r \in \aopt_i} \f_r(|(\aopt_i, \ne_{-i})|_r) \\
    &\geq \sum_{\ay, \xx, \bz} [(\ay+\xx) \f(|\ne|_r) - \\
    & \quad \ \xx \f_r(\ay + \xx) - \bz \f(\ay + \xx + 1)] \cdot \card{\rraxb} \\
    &\geq \frac{1}{D} \sum_{\ay, \xx, \bz} [\ay \f(\ay + \xx) - \bz \f(\ay + \xx + 1)] \paop \\
    &\geq 0.
\end{align*}

The first inequality comes from the definitions of the utility function. The second inequality comes from counting the resources that are selected in the either $\ne_i$ or $\aopt_i$ by the agent in each set of resources in $\rraxb$. The third inequality arises from the fact that $|\ne|_r \leq \ay+\xx$, and since $\f$ is assumed to be non-increasing, $\f(|\ne|_r) \geq \f(\ay + \xx)$. The fourth inequality comes from the fact that since $\paop$ has to satisfy the inequality constraint in Eq. \eqref{eq:poaLP} to be feasible. Similarly, in a one-round walk, the best response for the first $\ay + \xx + \bz - 1$ is $\ne_i$. The best response for the other agents during the one-round walk is also $\ne_i$, since
\begin{align*}
    &\U(\ne_{j < i}, \ne_i, \emp_{j > i}) - \U(\ne_{j < i}, \aopt_i, \emp_{j > i}) \\
    &= \sum_{\ay, \xx, \bz} [\sum_{j=1}^{\ay+\xx} \f(i) - \xx \f(\ay + \xx) - \bz \f(\ay + \xx + 1)] \card{\rraxb} \\
    &\geq \frac{1}{D} \sum_{\ay, \xx, \bz} [\ay \f(\ay + \xx) - \bz \f(\ay + \xx + 1)] \paop \\
    &\geq 0.
\end{align*}

Therefore, the Nash equilibrium $\ne$ is reached from an empty configuration in one-round. Additionally, since $\ne$ is a Nash equilibrium, the resulting action state after $\k$ rounds can also be $\ne$. Therefore in this game, $\pob(\G; \k) \leq \poa(\G)$. Now we calculate the efficiency of the Nash equilibrium $\W(\ne)$ with respect to $\W(\aopt)$. We have that
\begin{align*}
    \W(\ne) = \sum_{\ay, \xx, \bz} \w(\ay + \xx) \cdot \paop \frac{N_2 - 2(\ay + \xx -1)}{N_2} + \\
    2 \sum_{i=1}^{\ay + \xx -1} \w(i) \frac{\paop}{N_2} = 1 + \mathrm{O}(\frac{1}{N_2}),
\end{align*}
where, since $\paop$ is feasible, then it satisfies the equality constraint that $\sum_{\ay, \xx, \bz} \w(\ay + \xx) \paop = 1$. $\mathrm{O}(\frac{1}{N_2})$ reflects that the rest of the terms are on order of $1/N_2$. Similarly, 
\begin{align*}
    \W(\aopt) = \sum_{\ay, \xx, \bz} \w(\bz + \xx) \cdot \paop \frac{N_2 - 3 (\bz + \xx -1)}{N_2} + \\
     2 \sum_{i=1}^{\bz + \xx -1} \w(i) \frac{\paop}{N_2} = \poa(\mathcal{G}_{\ww, \fw}^{N_1})^{-1} + \mathrm{O}(\frac{1}{N_2}),
\end{align*}
where, since $\paop$ is optimal, then $\sum_{\ay, \xx, \bz} \w(\bz + \xx) \paop = \poa(\mathcal{G}_{\ww, \fw}^{N_1})^{-1}$. For any $\varepsilon_2$, we can choose $N_2$, such that $\mathrm{O}(\frac{1}{N_2}) \leq \varepsilon_2$, so $\poa(\G) \leq \poa(\mathcal{G}_{\ww, \fw}^{N_1})^{-1} + \varepsilon_2$. To put everything together, we have that 
\begin{align*}
    &\pob(\setgm; \k) \leq \pob(\G; \k) \leq \poa(\G) \\
    &\leq \poa(\mathcal{G}_{\ww, \fw}^{N_1})^{-1} + \varepsilon_2 \leq \poa(\setgm) + \varepsilon_1 + \varepsilon_2,
\end{align*}
and since $\varepsilon_1$ and $\varepsilon_2$ are arbitrary, we have the result for any rounds $\k \geq 1$.
\end{proof}

\subsection{Proof of Theorem \ref{thm:poapobtradeoff}}

To characterize the Pareto optimal frontier in Eq. \eqref{eq:tradeoffsetcov}, we first simplify the linear program in Eq. \eqref{eq:dualtractLP} with respect to the set covering welfare rule.

\begin{lemma}
\label{lem:setpob}
Let $\ww = \{\wsc\}$, where $\wsc$ is the set covering welfare rule defined in \eqref{eq:wscdef}, and $\fw = \{\f\}$ be the corresponding utility rule. Then the one-round walk efficiency guarantee is
\begin{equation}
    \label{eq:pobwscf}
    \pob(\mathcal{G}_{\wsc, \f}; 1)^{-1} = \sum_{i \in \N}{\f(i)} - \min_{i \in \N}{\f(i)} + 1.
\end{equation}
\end{lemma}
\begin{proof}
Examine the dual program in Eq. \eqref{eq:dualtractLP} with substituting the set covering welfare defined in Eq. \eqref{eq:wscdef}. Under the substitution, the dual constraint for a given $\bz$, $\ay$ simplifies to 
\begin{equation*}
    \pbdual \geq \sum_{i=1}^{\ay} \f(i) - \bz \min_{1 \leq i \leq \ay+1} \f(i) + \min(1, \bz).
\end{equation*}
We have applied the fact $\wsc(j) = \min(1, j) = 1$ when $j \geq 1$ and $\rm{H} = \max _{j \in \N} \wsc(j)/j = 1$ to the dual constraint. Observe that the binding constraint occurs when we limit $\ay \to \infty$ and set $\bz = 1$ (and not $\bz=0$ since $\f(1) = 1$, the term $1 - \min_j \f(j) \geq 0$). Under those binding constraints, $\pob(\mathcal{G}_{\wsc, \f}; 1)^{-1} = \pbdual$, where $\pbdual$ matches the given expression in Eq. \eqref{eq:pobwscf}.
\end{proof}

\begin{figure}[ht]
    \centering
    \includegraphics[width=250pt]{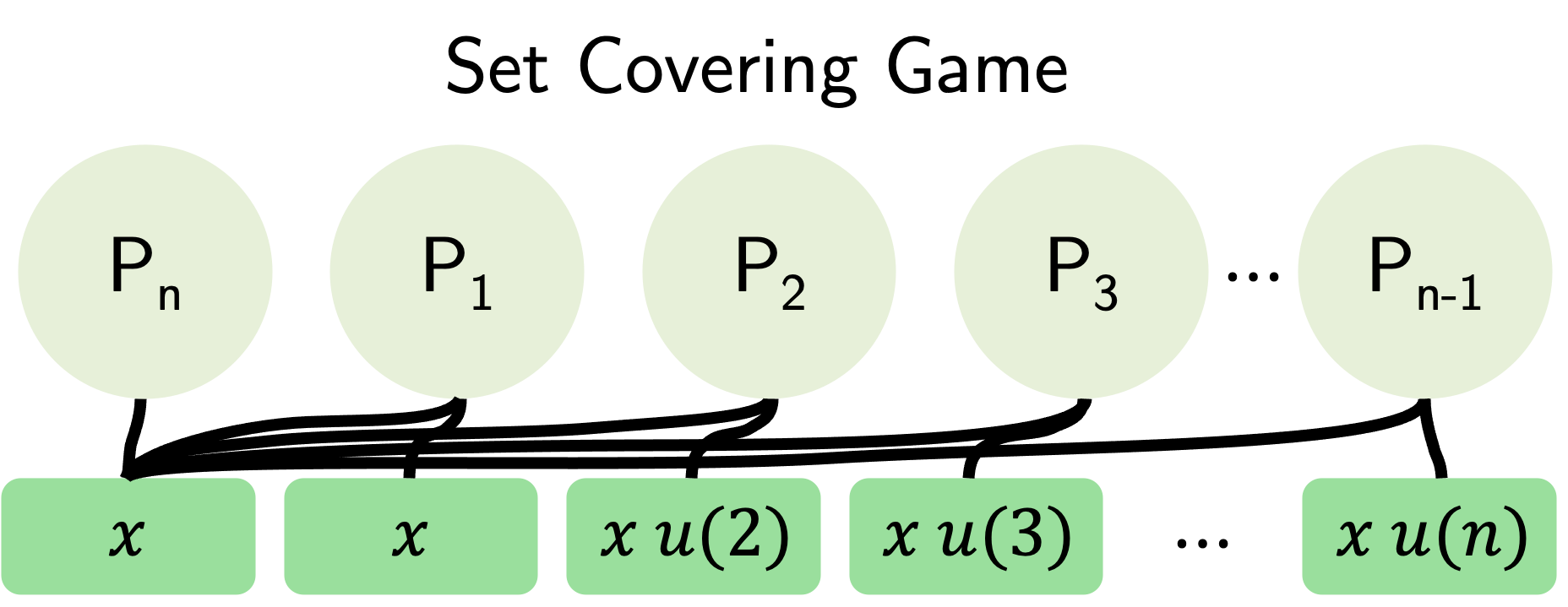}
    \caption{The worst case game construction achieving the one-round walk guarantee dictated by Lemma \ref{lem:setpob}. In this game, all the agents can either stack on the first resource set or spread out.}
    \label{fig:setcovgame}
\end{figure}

To characterize the trade-off, we now provide an explicit expression of Pareto optimal utility rules, i.e., the utility rules $\f$ that satisfy either $\pob(\mathcal{G}_{\wsc, \f}; 1) \geq \pob(\mathcal{G}_{\wsc, \f'}; 1)$ or $\poa(\mathcal{G}_{\wsc, \f}) \geq \poa(\mathcal{G}_{\wsc, \f'})$ for all $\f'\neq \f$.

\begin{lemma} \label{lem:implicitfboundar}
For a given $\ipoa \geq 0$, a utility rule $
\fxx$ that satisfies $\poa(\mathcal{G}_{\wsc, \f}) \geq 1/(1+\ipoa)$ while maximizing $\pob(\mathcal{G}_{\wsc, \f}; 1)$ is defined as in the following recursive formula:
\begin{align}
    \fxx(1) &= 1     \label{eq:recurf} \\
    \fxx(j+1) &= \max\{j \fxx(j) - \ipoa, 0\}. \nonumber
\end{align}
\end{lemma}
\begin{proof}
According to \cite[Theorem 2]{paccagnan2018utility}, the price of anarchy can be written as 
\begin{align}
    \frac{1}{\poa(\mathcal{G}_{\wsc, \f})} &= 1 + \max_{1 \leq j \leq n - 1} \{(j+1) \f(j+1) - 1, \nonumber \\
    &j\f(j)-\f(j+1), j\f(j+1)\} . \label{eq:poasetgen}
\end{align}
We first show that if $\f$ is Pareto optimal, then it must also be non-increasing. Otherwise, we show that another $\f'$ exists that achieves at least the same one round efficiency, but a higher price of anarchy, contradicting our assumption that $\f$ is Pareto optimal. Assume, by contradiction, that there exists $\f$ that is Pareto optimal and not non-increasing, i.e., there exists a $J \geq 1$, in which $\f(J) < \f(J+1)$. Notice that switching the value $\f(J)$ with $\f(J+1)$ results in an unchanged one round efficiency according to Eq. \eqref{eq:pobwscf} in Lemma \ref{lem:setpob} if $J > 1$. We show that $\f'$ with the values at $J$ and $J+1$ switched has a higher price of anarchy than $\f$.

For any $1 \leq J \leq n - 1$, the expressions from Eq. \eqref{eq:poasetgen} that include $\f(J)$ or $\f(J+1)$ are
\begin{align*}
    J \f(J), \quad (J+1) \f(J+1), \quad (J-1) \f(J-1) - \f(J) + 1, \\
    J \f(J) - \f(J+1) + 1, \quad (J+1) \f(J+1) - \f(J+2) + 1, \\
    (J-1) \f(J) + 1, \quad J \f(J+1) + 1.
\end{align*}
After switching, the relevant expressions for $\f'$ are
\begin{align*}
    J \f(J+1), \ (J+1) \f(J), \  (J-1) \f(J-1) - \f(J+1) + 1, \\
    J \f(J+1) - \f(J) + 1, \quad (J+1) \f(J) - \f(J+2) + 1, \\
    (J-1) \f(J+1) + 1, \quad J \f(J) + 1.
\end{align*}
Since $\f(J) < \f(J+1)$, switching the values results in a strictly looser set of constraints, and the value of the binding constraint in Eq. \eqref{eq:poasetgen} for $\f'$ is at most the value of the binding constraint for $\f$. Therefore $\poa(\f) \leq \poa(\f')$. Note that if $J= 1$, switching $J$ and $J+1$ and scaling down appropriately so $\f'(1)=1$, then $\pob(\f') > \pob(\f)$ as well. This contradicts our assumption that $\f$ is Pareto optimal.

Now we restrict our focus $\f$ that are non-increasing. Under this assumption, the price of anarchy is
\begin{equation*}
    \frac{1}{\poa(\wsc, \f)} = 1 + \\ \max_{1 \leq j \leq n - 1} \{j\f(j)-\f(j+1), (n-1)\f(n)\},
\end{equation*}
as detailed in Corollary $2$ in \cite{paccagnan2018utility}. Let 
\begin{equation}
\label{eq:chiconst}
    \ipoa_{\f} = \max_{1 \leq j \leq n - 1} \{j\f(j)-\f(j+1), (n-1)\f(n)\}.
\end{equation} 
For $\f$ to be Pareto optimal, we claim that $ j \f(j) - \f(j+1) = \ipoa_{\f}$ must hold for all $j$. Consider any other $\f'$ with $\ipoa_{\f} = \ipoa_{\f'}$. It follows that $\poa(\f) = \poa(\f') = 1/(1+\ipoa_{
\f})$. By induction, we show that $\f(j) \leq \f'(j)$ for all $j$. The base case is satisfied, as $1 = f(1) \leq f'(1) = 1$. Under the assumption $\f(j) \leq \f'(j)$, we also have that 
\begin{equation}
    j\f(j) - \ipoa_{\f} = \f(j+1) \leq \f'(j+1) = j\f'(j) - \ipoa^j_{\f},
\end{equation}
where $\ipoa^j_{\f'} = j\f'(j) - \f'(j+1) \leq \ipoa_{\f'}$ by definition in Eq. \eqref{eq:chiconst}, and so $\f(j) \leq \f'(j)$ for all $j$. Therefore the summation $\sum_{i \in \N}{\f(i)} - \min_{i \in \N}{\f(i)}$ in Eq. \eqref{eq:pobwscf} is diminished and $\pob(\f) \geq \pob(\f')$, proving our claim. As $\f$ must satisfy $\f(j) \geq 0$ for all $j$ to be a valid utility rule, $\f(j+1)$ is set to be $\max \{j\f(j) - \ipoa, 0\}$. Then we get the recursive definition for the maximal $\fxx$. Finally, we note that for infinite $n$, $\ipoa \leq \frac{1}{e-1}$ is not achievable, as shown in \cite{gairing2009covering}.
\end{proof}

\noindent With the two previous lemmas, we can move to proving Theorem \ref{thm:poapobtradeoff} We first characterize a closed form expression of the maximal utility rule $\fxx$, which is given in Lemma \ref{lem:implicitfboundar}. We fix $\ipoa$ so that $\poa(\fxx) = \frac{1}{\ipoa+1}$ = $Q$. To calculate the expression for $\fxx$ for a given $\ipoa$, a corresponding time varying, discrete time system to Eq. \eqref{eq:recurf} is constructed as follows.
\begin{align*}
    x(t+1) &= t x(t) - \ipoa, \\
    y(t) &= \max \{x(t), 0\}, \\
    x(1) &= 1,
\end{align*}
where $y(t) \equiv \fxx(j)$ corresponds to the utility rule. Solving for the explicit solution for $y(t)$ using the state-transition matrix gives
\begin{align*}
    y(1) &= 1 \\
    y(t) &= \max \Big[ \prod_{\ell=1}^{t-1} \ell - \ipoa \big( \sum_{\tau=1}^{t-2} \prod_{\ell=\tau+1}^{t-1} \ell \big) - \ipoa , 0 \Big] \ \ t > 1.
\end{align*}
Simplifying the expression and substituting for $\fxx(j)$ gives 
\begin{equation*}
    \fxx(j) = \max \Big[ (j-1)!(1 - \ipoa \sum_{\tau=1}^{j-1} \frac{1}{\tau!} \big), 0 \Big] \ \ \ j \geq 1.
\end{equation*}

Substituting the expression for the maximal $\fxx$ into Eq. \eqref{eq:pobwscf} gives the one round efficiency. Notice that for $\ipoa \geq \frac{1}{e-1}$, $\lim_{j \to \infty} \fxx(j) = 0$, and therefore $\min_j{\fxx(j)} = 0$. Shifting the variables $j' = j + 1$, we get the statement in Eq. \eqref{eq:tradeoffsetcov}.

\subsection{Supermodular welfare rules}
In this section, efficiency of one-round walks are examined for classes of \emph{supermodular games}. Supermodular games are an important sub-class of resource allocation games, in which there is a surplus of added system welfare when a resource is covered by more than one agent. Applications of supermodular games include clustering and power allocation in networks \cite{paccagnan2021utility}. A welfare rule $\w$ is deemed to be supermodular if $\w(j) - \w(j-1)$ is increasing and non-negative for all $j \geq 1$. Interestingly, for supermodular games, the utility designs that both maximize the one-round efficiency and price of anarchy include the constant utility design $\fw(\w) = \f_{\ones}$ in which $\f_{\ones}(j) = 1$ for all $j \in \N$, and the \emph{Shapley} utility design $\fw(\w) = \f_{\rm{shap}}$ in which $\f_{\rm{shap}}(j) = \w(j)/j$ for all $j \in \N$. Furthermore, the optimal one-round and price of anarchy guarantees are equivalent, as seen below.

\begin{prop}
\label{prop:supmodopt}
Consider a set of supermodular welfare rules $\ww = \{\w^1, \dots, \w^m\}$ with $\w^{\ell}(1) = 1$. If the number of agents is fixed to $n$, then the optimal one-round and price of anarchy guarantees are as follows
\begin{equation}
    \sup_{\fw} \pob(\setgmn; 1) = \sup_{\fw} \poa(\setgmn) = \min_{1 \leq \ell \leq m} \frac{n}{\w^{\ell}(n)}.
\end{equation}
Furthermore any utility design in which $\f^{\ell} = \fw(\w^{\ell})$ is non-decreasing and satisfies $\f^{\ell}(1) = 1$ and $\sum_{i=1}^{j} \f^{\ell}(i)/w(j) \leq 1$ for all $1 \leq j \leq n$ achieves the optimal one round efficiency guarantee.
\end{prop}
\begin{proof}
The fact that $\sup_{\fw} \poa(\setgmn) = \min_{1 \leq \ell \leq m} \frac{n}{\w^{\ell}(n)}$ comes from applying the result in \cite[Theorem 4]{paccagnan2021utility} to a class of welfare rules. Thus, we first show that $\sup_{\fw} \pob(\setgmn; 1) \leq \min_{1 \leq \ell \leq m} \frac{n}{\w^{\ell}(n)}$ through a game construction that is valid for any utility design $\fw$. Let $\w^* = \arg \min_{1 \leq \ell \leq m} \frac{n}{\w^{\ell}(n)}$ be the welfare rule that attains the minimum. Let the game $\G$ have $n$ agents with agent $i$ having the action set $\ac_i = \{\emp_i, \abr_i, \aopt_i \}$. There are $n+1$ resources which are all endowed with the welfare rule $\w_r = \w^*$ for all $r \in \rr$, with agent $i$ either selecting $\abr_i = \{r_{i+1}\}$ or $\aopt_i = \{r_1\}$. Under any utility rule $\f$, each agent $i$ is indifferent to choosing $\abr_i$ or $\aopt_i$ if no other agents $j \neq i$ have selected $r_1$ through $\aopt_j$. Thus $\abr$ can result after a $\k$-round walk with a welfare of $\W(\abr) = n$. The welfare of the optimal allocation $\aopt$ is $\W(\aopt) = \w(n)$. Therefore, for any utility design $\f = \fw(\w^*)$, the efficiency is bounded by $\pob(\setgmn; 1) \leq \pob(\G; 1) = \frac{n}{\w^*(n)}$. We remark that $\pob(\setgmn; \k) \leq \pob(\G; \k) = \frac{n}{\w^*(n)}$ as well, since $\abr$ is a Nash equilibrium.

Now we show that for a utility design $\fw$, such that the utility rule $\f^{\ell} = \fw(\w^{\ell})$ is non-decreasing and satisfies $\sum_{i=1}^{j} \f^{\ell}(i)/w(j) \leq 1$ and $\f^{\ell}(1) = 1$ for every $j$ and $\ell$, the one-round efficiency is lower bounded by $\pob(\setgmn; 1) \geq \min_{1 \leq \ell \leq m} \frac{n}{\w^{\ell}(n)}$. To do this, we can use a modified version of the linear program in Eq. \eqref{eq:dualtractLP} in which $\pob(\setgmn; 1) \geq \min_{1 \leq \ell \leq m} \frac{1}{\pbdual^{\ell}}$, where $\pbdual^{\ell} \in \R$ is the solution to
\begin{align*}
& \pbdual^{\ell} = \min \quad \pbdual \quad \textrm{subject to:} \\
& \pbdual \w^{\ell}(\ay) \geq \rm{H}^{\ell} \left( \sum_{i=1}^\ay \f^\ell(i) - \bz \min_{1 \leq i \leq \ay+1} \f^\ell(i) \right) + \w^\ell(\bz) \\
&\text{for all } 0 \leq \bz \leq n \text{ and } 1 \leq \ay \leq n,
\end{align*}
where the linear program is a lower bound since we consider tighter constraints that allow $\ay$ and $\bz$ to to range from $1$ to $n$. Since $\w^{\ell}$ is supermodular, $\rm{H}^{\ell} = \w^{\ell}(n)/n$ and assuming $\f^{\ell}$ is non-decreasing, $\min_{1 \leq i \leq \ay+1} \f^\ell(i) = \f^{\ell}(1) = 1$. Thus, we can simplify the constraint as
\begin{equation}
    \pbdual \w^{\ell}(\ay) \geq \frac{\w^{\ell}(n)}{n} \sum_{i=1}^\ay \f^\ell(i) - \frac{\w^{\ell}(n)}{n} \bz + \w^\ell(\bz)
\end{equation}

With this, we observe that $\w^{\ell}(\bz) - \bz \cdot \w^{\ell}(n)/n$ is convex in $\bz$. So the binding constraint for $\bz$ occurs at either the end point $\bz = 0$ or $\bz = n$ and the terms can be cancelled out. Additionally, $\max_{\ay} \sum_{i=1}^\ay \f^\ell(i) / \w^{\ell}(\ay) = 1$ occurs at the binding constraint $\ay = 1$, by assumption that $\sum_{i=1}^{j} \f^{\ell}(i)/w(j) \leq 1$ for all $1 \leq j \leq n$. Therefore, $\pbdual^{\ell} = \rm{H}^{\ell} = \w^{\ell}(n)/n$ under the binding constraint of $\ay = 1$ and $\bz = 0$ (or $n$) and we indeed have that $\pob(\setgmn; 1) \geq \min_{1 \leq \ell \leq m} \frac{n}{\w^{\ell}(n)}$.
\end{proof}

We remark that the utility rules $\f_{\ones}$ and $\f_{\rm{shap}}$ both satisfy the assumptions in Proposition \ref{prop:supmodopt}. Furthermore, we remark that the optimal one-round guarantees match the optimal $\k$-round and price of anarchy guarantees, and so running the best-response process for further rounds does not increase the resulting efficiency guarantees.

\begin{figure}[ht]
    \centering
    \includegraphics[width=250pt]{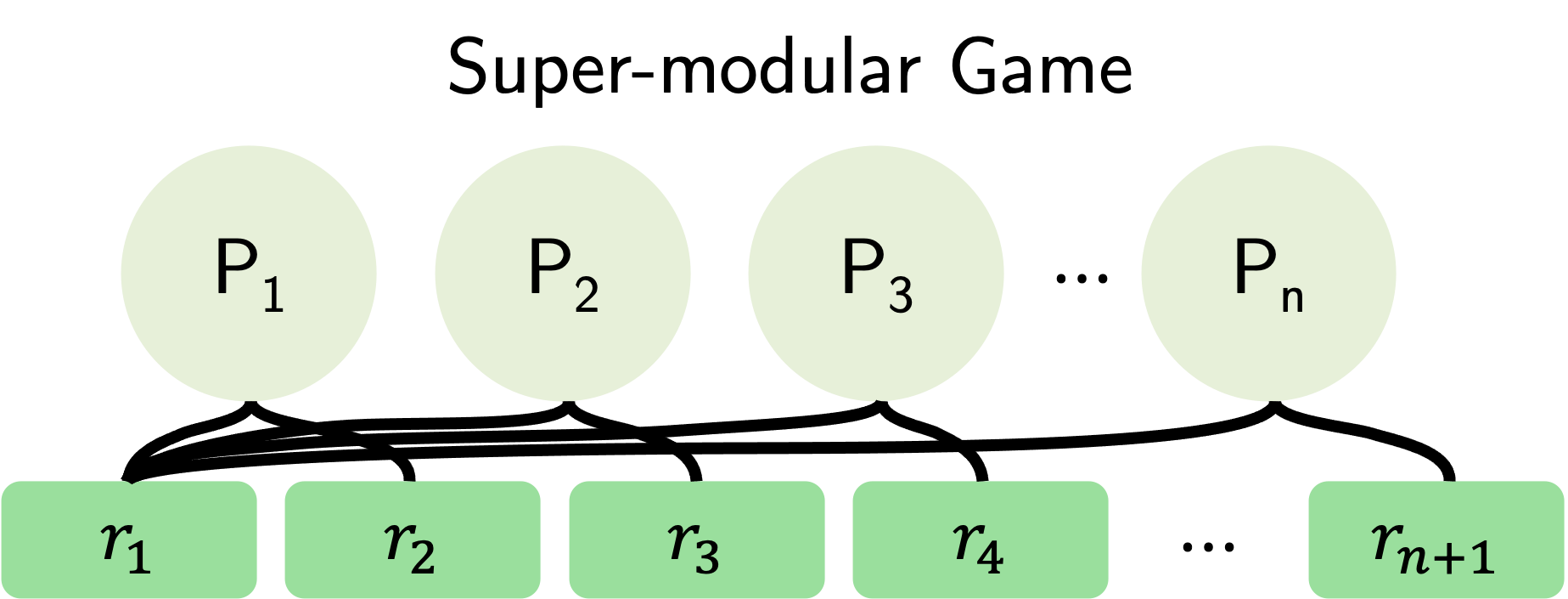}
    \caption{The worst case game construction used in Proposition \ref{prop:supmodopt}. If the previous agents $j \leq i$ choose not to stack, each agent $i$ is indifferent for its two actions. Therefore, the worst allocation that can result from a $\k$-round walk is all the agents not choosing to stack.
    \label{fig:supermodgame}}
\end{figure}

\begin{IEEEbiography}[{\includegraphics[width=1in,height=1.25in,clip,keepaspectratio]{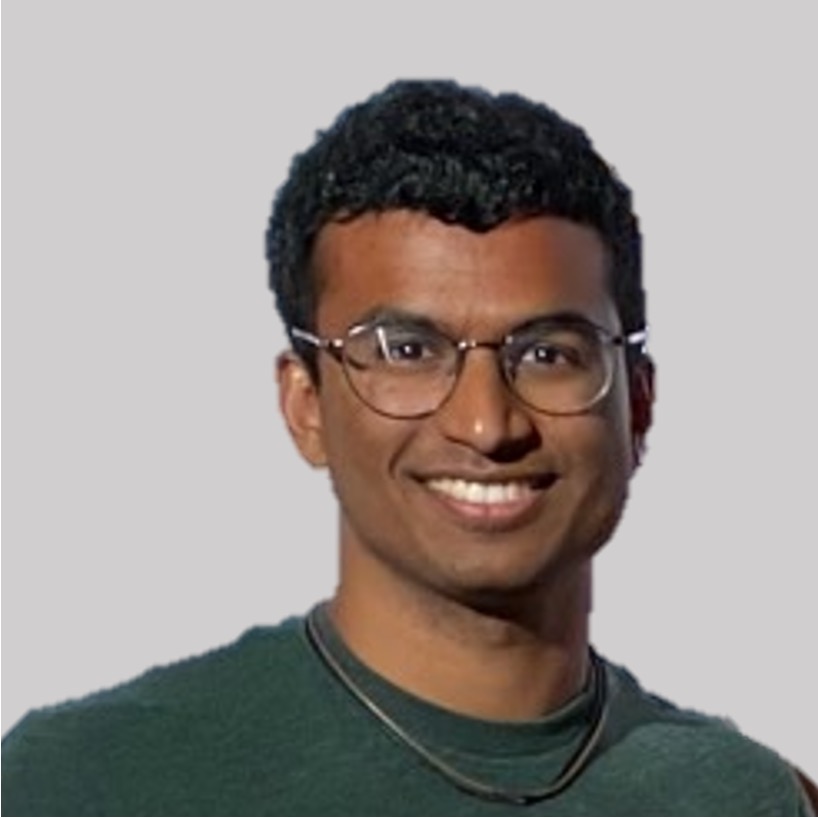}}]{Rohit Konda} received a B.S in Biomedical Engineering in 2018 and an M.S in Electrical and Computer Engineering in 2019 from Georgia Tech. He is currently a Regent's Fellow pursuing a Ph.D in the Department of Electrical and Computer Engineering at the University of California, Santa Barbara. His research interests include distributed optimization and game-theoretic applications in multi-agent systems.
\end{IEEEbiography}

\begin{IEEEbiography}[{\includegraphics[width=1in,height=1.25in,clip,keepaspectratio]{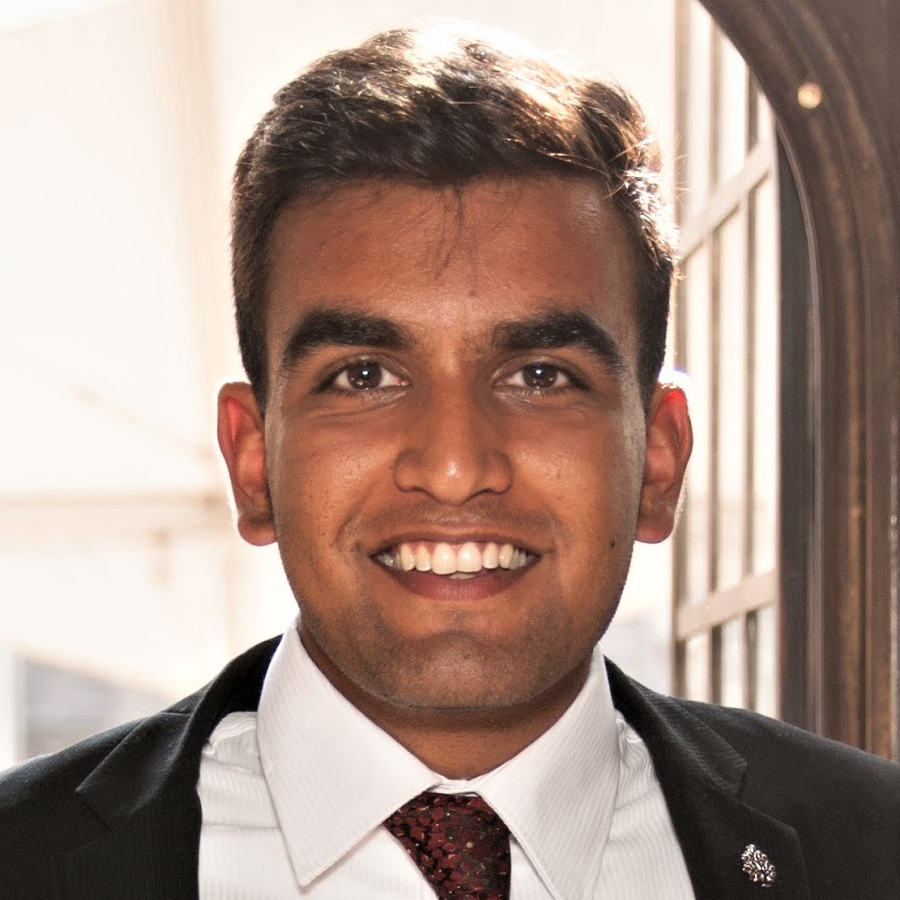}}]{Rahul Chandan} received his B.A.Sc. degree in Engineering Science, and Electrical and Computer Engineering from the University of Toronto, Toronto, ON, Canada, in June 2017. He is currently working toward a Ph.D. degree in the Electrical and Computer Engineering Department, University of California, Santa Barbara, CA, USA, since September 2017. His research interests include the application of game theoretic and classical control methods to the analysis and control of multi-agent systems.
\end{IEEEbiography}

\begin{IEEEbiography}[{\includegraphics[width=1in,height=1.25in,clip,keepaspectratio]{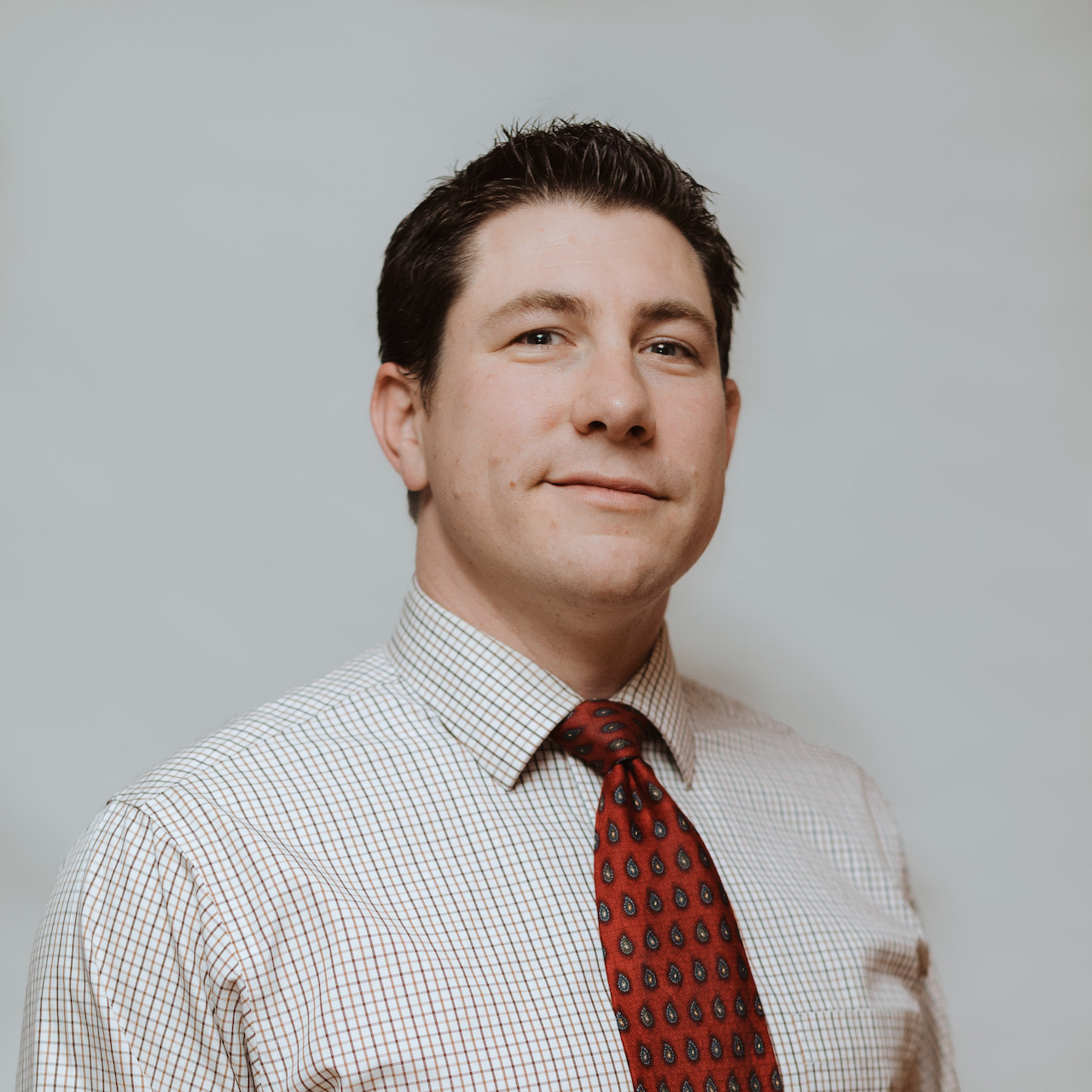}}]{David Grimsman} received the B.S. degree in electrical and computer engineering from Brigham Young University, Provo, UT, USA, in 2006 as a Heritage Scholar, with a focus on signals and systems. He received the M.S. degree in computer science from Brigham Young University in 2016. He received the Ph.D. degree in 2021, under the supervision of Jason R. Marden and Joao P. Hespanha, with the Electrical and Computer Engineering Department, University of California, Santa Barbara, Santa Barbara, CA, USA. He is currently an Assistant Professor with the Department of Computer Science in Brigham Young University.

He worked with BrainStorm, Inc. for several years as a Trainer and
IT Manager. His research interests include multiagent systems, game
theory, distributed optimization, network science, linear systems theory, and security of cyberphysical systems.
\end{IEEEbiography}

\begin{IEEEbiography}[{\includegraphics[width=1in,height=1.25in,clip,keepaspectratio]{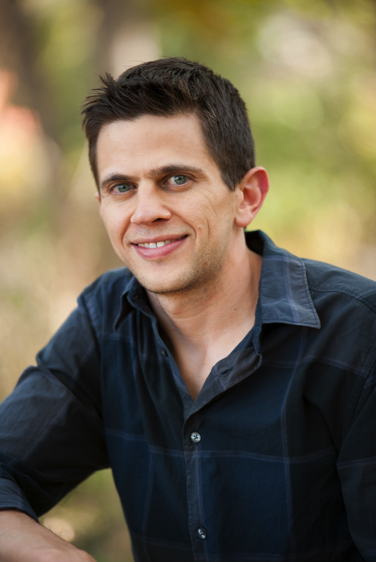}}]{Jason R. Marden} received the B.S. degree in 2001 and the Ph.D. degree in 2007 under the supervision of Jeff S. Shamma from the University of California, Los Angeles (UCLA), Los Angeles, CA, USA, both in mechanical engineering. After graduating from UCLA, he served as a Junior Fellow with the Social and Information Sciences Laboratory, California Institute of Technology until 2010 when he joined the University of Colorado. He is currently a Full Professor with the Department of Electrical and Computer Engineering, University of California, Santa Barbara, CA, USA. His research interests focus on game theoretic methods for the control of distributed multi-agent systems.

Dr. Marden is a recipient of the NSF Career Award in 2014, he was awarded the Outstanding Graduating Ph.D. degree, the ONR Young Investigator Award in 2015, the AFOSR Young Investigator Award in 2012, the American Automatic Control Council Donald P. Eckman Award in 2012, and the SIAG/CST Best SICON Paper Prize in 2015.
\end{IEEEbiography}

\end{document}